\documentclass[11pt, reqno, a4paper]{amsart}
\usepackage{amsmath,amssymb,amsthm}
\usepackage[margin=0.5in]{geometry}
\usepackage{graphicx}
\usepackage{xcolor}
\usepackage{bbm}
\usepackage{enumitem}
\allowdisplaybreaks
\DeclareMathOperator{\Span}{span}
\DeclareMathOperator{\Len}{len}
\DeclareMathOperator{\Diag}{diag}
\DeclareMathOperator{\Dim}{dim}
\DeclareMathOperator*{\argmin}{arg\,min}
\DeclareMathOperator*{\argmax}{arg\,max}
\DeclareMathOperator*{\esssup}{ess\,sup}
\allowdisplaybreaks
\newtheorem{theorem}{Theorem}[section]
\newtheorem{lemma}[theorem]{Lemma}

\usepackage{xcolor}
\usepackage{bigints}
\newcommand\numberthis{\addtocounter{equation}{1}\tag{\theequation}}
\theoremstyle{definition}
\newtheorem{definition}[theorem]{Definition}

\newtheorem{prop}[theorem]{Proposition}
\DeclareMathOperator\supp{supp}

\newcommand\norm[1]{\left\lVert#1\right\rVert}
\theoremstyle{remark}
\newtheorem{remark}[theorem]{Remark}
\DeclareMathOperator{\Z}{\mathbb{Z}}

\numberwithin{equation}{section}

\begin{document}
	
	\title{Optimal shift-invariant spaces from uniform measurements}
	
	%Information for first author
	\author{Rohan Joy}
	%Address of record for the research reported here
	\address{Department of Mathematics, Indian Institute of Technology Madras, India}
	\email{rohanjoy96@gmail.com}
	% \thanks will become a 1st page footnote.
	%\thanks{The first author was supported in part by NSF Grant \#000000.}
	
	%  Information for second author
	% \author{Felix Krahmer}
	% \address{Center of Mathematics, Technical University of Munich, Germany}
	% \email{felix.krahmer@tum.de}
	% %\thanks{Support information for the second author.}

 %        \author{Alessandro Lupoli}
	% \address{Center of Mathematics, Technical University of Munich, Germany}
	% \email{alessandro.lupoli@tum.de}
 
	%  Information for third author
	\author{Radha Ramakrishnan}
	\address{Department of Mathematics, Indian Institute of Technology Madras, India}
	\email{radharam@iitm.ac.in}
	%\thanks{Support information for the second author.}

	%    General info
	
	%	\date{January 1, 2001 and, in revised form, June 22, 2001.}

\begin{abstract} Let $m$ be a positive integer and 
 $\mathcal{C}$ be a collection  of closed subspaces  in $L^2(\mathbb{R})$. Given the measurements $\mathcal{F}_Y=\left\lbrace  \left\lbrace  y_k^1 \right\rbrace_{k\in \mathbb{Z}},\ldots, \left\lbrace  y_k^m \right\rbrace_{k\in \mathbb{Z}}  \right\rbrace \subset \ell^2(\mathbb{Z})$ of unknown functions $\mathcal{F}=\left\{f_1, \ldots,f_m \right\} \subset L^2( \mathbb{R})$, in this paper we study the problem of finding an optimal space $S$ in $\mathcal{C}$ that is ``closest" to the measurements $\mathcal{F}_Y$ of $\mathcal{F}$. Since the class of finitely generated shift-invariant spaces (FSISs) is popularly used for modelling signals, we assume $\mathcal{C}$ consists of FSISs. We will be considering three cases. In the first case, $\mathcal{C}$ consists of FSISs without any assumption on extra invariance. In the second case, we assume $\mathcal{C}$ consists of extra invariant FSISs, and in the third case, we assume $\mathcal{C}$ has translation-invariant FSISs. In all three cases, we prove the existence of an optimal space.

 \end{abstract}
	\maketitle
% For  any subset $\phi_1,\cdots, \phi_l\in L^2(\mathbb{R})$ we define
% \[V(\phi_1,\cdots,\phi_l)=\overline{\Span} \left\{T_k \phi_i: k \in \mathbb{Z}, i \in \{1, \cdots, l\} \right\}.\]

% For a finitely generated SIS $V \subset L^2(\mathbb{R})$ we define the length of $V$ as 
% \[\Len(V)=\min \left\{ n \in \mathbb{N}: \exists
%  \phi_1, \cdots, \phi_n \in V \text{ with } V= V(\phi_1,\cdots, \phi_n) \right\}.\]
\section{Introduction}\label{intro}
Let $\mathcal{C}$ be a family of closed subspaces of $L^2(\mathbb{R})$, and $\mathcal{F}=\left\lbrace f_1,f_2,\ldots,f_m\right\rbrace$, a finite set of elements in $L^2(\mathbb{R})$. In this article, we study the problem of finding an optimal space $S$ in class $\mathcal{C}$ that is ``closest'' to the ``measurements'' of the functions of $\mathcal{F}$. The primary objective in identifying an appropriate space is to fit the space to the given data, rather than modifying the data to conform to existing models. This is crucial because signal acquisition often introduces noise, causing inherently low-dimensional signals to appear high-dimensional. Therefore, the aim is to accurately identify the original low-dimensional space in which these signals reside.

In most real-life applications, such as digital signal and image processing, signals and images are generally assumed to belong to finitely generated shift-invariant spaces (FSISs) of the form 
\begin{equation}\label{FSISform}
V(\phi_1, \ldots, \phi_l):=\overline{\Span} \left\{ \phi_i(\cdot-k): i \in \{1, \ldots, l\}, k \in \mathbb{Z}, \text{ and } \phi_1, \ldots, \phi_l \in L^2(\mathbb{R}) \right\}.
\end{equation} 
The functions ${\phi_1, \ldots, \phi_l}$ are called the generators of the space $V(\phi_1, \ldots, \phi_l)$. Hence, in this paper, we study the case where the approximation subspaces (the collection $\mathcal{C}$) consist of FSISs.

Our work is motivated by the original data approximation problem proposed by Aldroubi et. al. in \cite{aldroubiacha2007} and by the series of subsequent works \cite{aldroubijfaa2008,aldroubicollectmath2012,aldroubifca2011,barbierijmpa2020,cabrelliacha2016}.
In \cite{aldroubiacha2007}, the authors posed and answered positively the following question. Given a large set of experimental data $\{f_1, \ldots, f_m \} \subset L^2(\mathbb{R})$, does there exist a minimizer to the problem
\[\argmin_{V \in \mathcal{V}_n} \sum_{i=1}^m \left\| f_i - P_V f_i \right\|^2 ?\]
Here, $\mathcal{V}_n$ consists of all FSISs with at most $n$ generators. We wish to explore the above problem from a sampling theory perspective, while at the same time also considering other popular classes of FSISs. Fix $n_0 \in \mathbb{N}$. We assume that instead of functional data $\{f_1, \ldots, f_m \}$, the measurements (taken using a sampling operator) $\mathcal{F}_Y=\left\lbrace  \left\lbrace  y_k^1 \right\rbrace_{k\in \mathbb{Z}},\ldots, \left\lbrace  y_k^m \right\rbrace_{k\in \mathbb{Z}}  \right\rbrace \subseteq \ell^2(\mathbb{Z})$ of the functions $\{f_1, \ldots, f_m \}$  on the uniform grid $\frac{\mathbb{Z}}{n_0}$ are given to us. Our aim is to search for an optimal space that is nearest to this observed data. For this, we define an appropriate  \textbf{minimization problem}. The problem is divided into two parts:
\begin{enumerate}
\item  The first step is to find a good approximation of the unknown functions $\{f_1, \ldots,f_m \}$ from the measurements $\mathcal{F}_Y$. Since we do not presume any kind of rich data condition, we use the following extremely popular reconstruction algorithm from learning theory \cite{smaleacha2005}.
Fix $\lambda > 0$. Pick a $V \in \mathcal{C}$, and for each $j \in \lbrace 1,\ldots,m\rbrace$, find a function in $V$ (if it exists) whose measurements are the best least square regularized estimate for the given data $\left\lbrace y_k^j \right\rbrace_{k \in \mathbb{Z}}$. In other words, find
\begin{align*}
% & \arg\min_{f\in V} \norm{ \left\lbrace y_k^j \right\rbrace_{k \in \mathbb{Z}} - S_g^{n_0}(f) }_{\ell^2(\mathbb{Z})}^2+\lambda \norm{f}^2\\
% & \arg\min_{f\in V} \sum_{k\in \mathbb{Z}} \left| y_k^{j}-\left\langle f,g\left( -\frac{k}{n_0} \right)  \right\rangle\right|^2 + \lambda \norm{f}^2\\
\numberthis \label{mainproblemintroducedthefirsttime}& \arg\min_{f\in V} \left\{\sum_{k\in \mathbb{Z}} \left| y_k^j - f\left( \frac{k^g}{n_0}\right)\right|^2 + \lambda \norm{f}^2\right\}.
\end{align*}

Here, $\left\{ f\left(\frac{k^g}{n_0}\right) \right\}_{k \in \mathbb{Z}} \in \ell^2(\mathbb{Z})$ represents the measurements of the function $f$, taken using a sampling operator (that involves $g \in L^2(\mathbb{R})$) on the uniform grid $\frac{\mathbb{Z}}{n_0}$. The precise definition of this sampling operator will be provided later in Section \ref{problemsetup}.
\item The second step is to vary $V\in \mathcal{C}$ and find an optimal subspace $S \in \mathcal{C}$ which minimizes the error in \eqref{mainproblemintroducedthefirsttime}  when summed over all $j \in \{1, \ldots,m\}$. That is, find
\begin{equation}
\arg\min_{V \in \mathcal{C}}\sum_{j =1} ^m\left(\min_{f\in V} \sum_{k\in \mathbb{Z}} \left| y_k^j - f\left( \frac{k^g}{n_0}\right)\right|^2 + \lambda \norm{f}^2\right).\\
\end{equation}

\end{enumerate}

\smallskip

To address this new minimization problem, we utilize the following space. Fix $\lambda > 0$, and define 
\[
\widetilde{L^2(\mathbb{R})} := \left\{ \left( \left\{ f\left(\frac{k^g}{n_0}\right) \right\}_{k \in \mathbb{Z}}, f \right) : f \in L^2(\mathbb{R}) \right\}
\]
endowed with the norm
\[
\left\| \left( \left\{ f\left(\frac{k^g}{n_0}\right) \right\}_{k \in \mathbb{Z}}, f \right) \right\|^2 := \left\| \left\{ f\left(\frac{k^g}{n_0}\right) \right\}_{k \in \mathbb{Z}} \right\|^2_{\ell^2(\mathbb{Z})} + \lambda \left\| f \right\|^2_{L^2(\mathbb{R})}.
\]

The objective is to construct a space in which the elements consist of signals paired with their measurements. The norm of an element $\left( \left\{ f\left(\frac{k^g}{n_0}\right) \right\}_{k \in \mathbb{Z}}, f \right)$ simultaneously considers both the norm of the measurements of the function $f$ and the norm of the function $f$, weighted by a regularization parameter. This approach allows for the comparison of functions within our defined space using the measurements provided in the data, while also accounting for the norm of the approximating function. This method balances adherence to experimental data with the regularity of the function. Ideally, the minimizing function in this space should describe the given data accurately without being overly complex in its function norm, thereby effectively finding the regularized least square solution.

\smallskip

Given that we operate within this newly defined space $\widetilde{L^2(\mathbb{R})}$, and work with subspaces \[\widetilde{V} := \left\{ \left( \left\{ f\left(\frac{k^g}{n_0}\right) \right\}_{k \in \mathbb{Z}}, f \right) : f \in V \right\} \subset \widetilde{L^2(\mathbb{R})}\] generated using FSISs $V$, we develop a parallel theory of FSISs utilizing a newly defined fiber map, inspired by the classical one. This development is crucial because the classical theory does not adequately address our new setup.

\smallskip
Now, we introduce the three classes of approximation subspaces that are considered in our paper.
Fix $l \in \mathbb{N}$. Let \textbf{$n_0 \in \mathbb{N}$ be the measurement rate} at which the unknown functions $\{f_1, \ldots,f_m \}$ are sampled. 

First, the collection $\mathcal{C}$ is assumed to contain all FSISs that have at most $l$ generators. In this case, we are able to show the existence of an optimal space, but the problem of explicitly finding it is still open. The existence is shown by a straightforward application of \cite[Theorem 3.8]{aldroubifca2011}.

Next, we consider the case where the collection $\mathcal{C}$  contains FSISs that have extra invariance. Let $n \in   \mathbb{N}$. Then we say that $V=V(\phi_1, \ldots, \phi_l)$ is $\frac{\mathbb{Z}}{n}$-extra invariant if 
\[f \in V \implies T_{\frac{k}{n}}f \in V, \hspace{0.2cm} \forall \hspace{0.1cm} k \in \mathbb{Z}.\]
In this case, we present one of the main contributions of this paper. We show that when the collection $\mathcal{C}$ is assumed to contain $\frac{\mathbb{Z}}{n_0}$-extra invariant FSISs (recall that $n_0$ is our assumed measurement rate) having at most $l$ generators, then an optimal space exists, and we explicitly construct it.

Finally, we consider the minimization problem for the class $\mathcal{C}$ of FSISs with at most $l$ generators that are also translation invariant. We show that under the assumption that the translates of the generators of $V \in \mathcal{C}$ form a Riesz basis for $V$, an optimal space exists, and we explicitly construct it. The above class was introduced in \cite{cabrelliacha2016}. Our aim here is to explore whether, over the same class, our minimization problem can be solved.\\

In most real-life applications, we have measurements of signals rather than the signals themselves. Our goal is to explore the problem proposed by Aldroubi from a perspective that aligns more closely with these real-world scenarios. This specific form of the problem has not been studied previously. While Aldroubi et al. discussed the first case -where $\mathcal{C}$ is assumed to contain all finitely generated shift-invariant spaces (FSISs) with at most $l$ generators- in their paper \cite{aldroubifca2011}, our work rigorously addresses all three cases commonly found in the literature. We demonstrate the existence of an optimal space and show that complete knowledge of the function is unnecessary; the measurements alone are sufficient to establish the existence of an optimal space when the optimization problem is approached as described. This paper focuses on proving the existence of an optimal space and finding it explicitly, if possible, without delving into the error generated by the optimal space, which we leave for future work.

We remark here that the measurement rate $n_0$ is allowed to be greater than $1$ because, naturally, in a lot of cases, sampling at $n_0 = 1$, i.e., at $\mathbb{Z}$ will generally be insufficient when dealing with shift-invariant spaces that have more than one generator (for an example see in  \cite[Corollary 3.1]{radhanfao}). Thereby motivating the consideration of cases where $n_0$ exceeds one.
\smallskip

The rest of the paper is organized as follows. In Section \ref{prelim}, we state the relevant definitions and results from the literature that we require to solve our minimization problem. 
Section \ref{problemsetup} and Section \ref{fibermaptheory} are dedicated to the setting up of our problem in a mathematical rigorous form and for developing tools that will be used to prove our main results in the upcoming sections. Several generalizations of tools used in classical analysis are introduced and important technical lemmas are proved. In Section \ref{FSISs}, we present the case where the collection $\mathcal{C}$ consists of FSISs without any assumption on extra-invariance. In Section \ref{FSISwithextrainvariance}, we deal with FSISs which are $\frac{\mathbb{Z}}{n_0}$-extra invariant (recall that $n_0$ is our assumed sampling rate) and in Section \ref{palaeyr}, we deal with FSISs which are translation-invariant.

\smallskip
We remark that our approach to solving the minimization problems will be as follows. We will continually use the developed tools to reformulate the problem into progressively simplified forms. The specific formulation we choose will depend on the class of FSISs we are minimizing over. Throughout the paper, our overall technique is to adapt methods from the classical theory of data approximation for FSISs in such a way that both the measurements of the function and its norm are considered when minimizing the error rather than relying solely on the norm as in the classical case.

\section{Notation and Preliminaries}\label{prelim}
\begin{itemize}
\item For any Hilbert space $H$, let $\mathcal{B}(H)$ denote the space of bounded linear operators on $H$.
\item The cardinality of a finite set $A$ is denoted by $\#A$.
\item Let $H$ be a Hilbert space and $A,B$ be two closed subspaces of $H$. Then $A \dot{\oplus} B$ denotes the orthogonal direct sum of $A$ and $B$.
\item The sequence $\{e_l\}_{l \in \mathbb{Z}}$ denotes the standard orthonormal basis in $\ell^2(\mathbb{Z})$.
\item If $\mathcal{M}$ is a closed subspace of a Hilbert space $H$, then  $P_{\mathcal{M}}$ denotes the orthogonal projection operator of $H$ onto $\mathcal{M}$.
\item  For any $A \subset \mathbb{R}$, $\mathcal{X}_A$ denotes the characteristic function on $A$. 
\item Let $H$ be a Hilbert space. Then $\boldsymbol{0}$ denotes the zero vector in $H$. If it is not clear from the context what $H$ is, then we will explicitly specify it. \\
\end{itemize}

%     A set of closed subspaces $\mathcal{C}$ of a separable Hilbert space $\mathcal{H}$ has the Minimum subspace Approximation Property (MSAP) if for every finite subset  $ \boldsymbol{F} \subset \mathcal{H}$ there exists an element $V \in \mathcal{C}$ that minimizes the expression $ \sum_{f \in \boldsymbol{F}} \left\| f - P_Vf \right\|^2$ over all $V \in \mathcal{C}$. We will say that $\mathcal{C}$ has MSAP$(k)$ for some $k \in \mathbb{N}$ if the previous property holds for all subsets $\boldsymbol{F}$ of cardinality $m \leq k$.
% \end{definition}
% Let $\mathcal{H}$ denote a Hilbert space, and let $\pi$ be a continuous unitary representation of $\mathbb{Z}$, i.e. a continuous homomorphism $ \mathbb{Z} \longrightarrow \mathcal{B}(\mathcal{H})$.
The Fourier transform of any $f \in L^1(\mathbb{R})$  is defined as  \[\widehat{f}(\xi)= \int_{\mathbb{R}} f(x) e^{- 2 \pi i x \xi} dx,\quad\xi\in\mathbb{R}.\]
Since $L^1(\mathbb{R}) \cap L^2(\mathbb{R})$ is dense in $L^2(\mathbb{R})$, the Fourier transform can be extended to a unitary  operator on  $L^2(\mathbb{R})$. Let  $a \in \mathbb{R}$. Then for any $f \in L^2(\mathbb{R})$, the translation operator $T_a$ is defined as 
 \[T_a f(\cdot) =f(\cdot -a).\]
Note that $\widehat{(T_af)}(\xi)= e^{-2 \pi i a \xi} \widehat{f}(\xi),\quad\xi\in\mathbb{R}.$
\begin{definition}
		A sequence of functions $\{f_k\}_{k \in \mathbb{Z}}$ in a separable Hilbert space $H$ is said to be a Riesz basis for $H$, if there exist constants $0< A \leq B< \infty$ such that 
		\begin{equation}
			A\sum_{k \in \mathbb{Z}}|c_k|^2 \leq \bigg\|\sum_{k \in \mathbb{Z}}c_kf_k\bigg\|^2_{H} \leq 
			B\sum_{k \in \mathbb{Z}}|c_k|^2
		\end{equation}
		for all $\{c_k\}_{k \in \mathbb{Z}} \in \ell^2(\mathbb{Z})$,  and $H= \overline{\Span}\{f_k\}_{k \in \mathbb{Z}}$.
  %The constants $A$ and $B$ are known as lower and upper Riesz bounds, respectively.
	\end{definition}
\begin{definition}
        Let $f \in L^2(\mathbb{R})$. Then the sequence $ \left\{\widehat{f}(\xi +k) \right\}_{k \in \mathbb{Z}}$ belongs to $\ell^2(\mathbb{Z})$, for a.e. $\xi \in [0,1]$. Given an FSIS $V$ of $L^2(\mathbb{R})$ and $ \xi \in [0,1]$,  
    \[J_V(\xi):= \overline{ \left\{ \left\{\widehat{f}(\xi +k) \right\}_{k \in \mathbb{Z}} : f \in V  \right\}},\]
    where the closure is taken in the norm of $\ell^2(\mathbb{Z}).$ 
\end{definition}
    % Further, $V=V(\phi_1, \ldots, \phi_l)$ if and only if $J_V(\xi)= \Span \left\{ \left\{\widehat{\phi_1}(\xi +k) \right\}_{k \in \mathbb{Z}}, \ldots ,\left\{\widehat{\phi_l}(\xi +k) \right\}_{k \in \mathbb{Z}} \right\}.$
\begin{prop}
    Let $V_1, \dots,V_n$ be FSISs. If $V=V_1 \dot{\oplus} \cdots \dot{\oplus} V_n$, then  
    \begin{equation}
        J_V(\xi)=J_{V_1}(\xi) \dot{\oplus} \cdots \dot{\oplus} J_{V_n}(\xi), \quad \text{ for a.e. } \xi \in [0,1].
    \end{equation}
\end{prop}
\begin{definition}
The length of an FSIS $V \subset L^2(\mathbb{R})$ is defined as 
\[\Len V = \min \left\{ n \in \mathbb{N}: \exists \hspace{0.2cm} \phi_1, \ldots, \phi_n \in V \text{ with } V=V(\phi_1, \ldots, \phi_n) \right\}.\]
\end{definition}
The following theorem on the length of FSISs was proved by Boor et. al. in \cite{deboorjfa1994}.
\begin{theorem}\label{Boorformulaforlength}
Let $V$ be an FSIS. Then,
\begin{equation} \label{Boorformulaforlengthequation}
\Len V= \esssup_{\xi \in [0,1]} \Dim J_V(\xi).\\
\end{equation}
\end{theorem}

Given a fixed positive integer $n$, for each $k \in \{0, \ldots,n-1 \}$, we define the set $B_k$ \cite{aldroubijfaa2010} as
\begin{equation}\label{defintionofB_k}
    B_k=\cup_{j \in \mathbb{Z}} ([k,k+n] +nj).
\end{equation}
Note that each $B_k$ is $n\mathbb{Z}$-periodic, implicitly depends on $n$ and that collection $\left\{B_k \right\}_{k=0}^{n-1}$ partitions (up to sets of measure zero) the real line.

\smallskip

Given an FSIS $V \subset L^2(\mathbb{R})$, we associate the following subspaces:
\begin{equation}\label{defnVk}
    V_k=\left\{ f \in L^2(\mathbb{R}): \widehat{f}=\widehat{g}\mathcal{X}_{B_k} \text{ for some } g \in V \right\}, \quad k \in \{0,\ldots,n-1\}.
\end{equation}
The spaces $V_k$ are mutually orthogonal since the sets $B_k$ are disjoint (up to sets of measure zero). If $f \in L^2(\mathbb{R})$ and $ k \in \{0, \ldots, n-1\}$, then we let $f^k$ denote the function defined by
\[\widehat{f^k}= \widehat{f}\mathcal{X}_{B_k}.\]

 Letting $P_k$ denote the orthogonal projection of $L^2(\mathbb{R})$ onto $\{f \in L^2(\mathbb{R}): \supp(\widehat{f}) \subset B_k\}$, we get 
 \[V_k=P_k(V) \quad \text{and } \quad f^k=P_kf.\]
 Suppose $V=V(\phi_1, \ldots, \phi_l) \subset L^2(\mathbb{R})$. Then,  it can be shown that $V_k=V(\phi^k_1, \ldots, \phi_l^k)$ for each $k \in \{0, \ldots, n-1\}$. Hence, for a.e. $\xi \in [0,1],$
 \begin{equation}\label{JV_k}
 J_{V_k}(\xi)= \Span \left\{  \left\{ \widehat{\phi_i^k}(\xi+r) \right\}_{r \in \mathbb{Z}}  : i \in \{ 1, \ldots,l\} \right\}.
\end{equation}

\begin{theorem}\cite{aldroubijfaa2010}
Fix $n \in \mathbb{N}$. Let $V=V(\phi_1, \ldots, \phi_l) \subset L^2(\mathbb{R})$. Then, the following are equivalent.
\begin{enumerate}
\item $V$ is $\frac{\mathbb{Z}}{n}$-extra invariant.
\item $V_k \subset V$ for $k \in \{0, \ldots, n-1\}$.
\item If $f \in V$, then $f^k \in V$ for each $k \in \{0, \ldots, n-1\}$.
\item $\left\{ \phi_i^k: i \in \{1, \ldots,l\} \right\} \subset V$ for each $k \in \{0, \ldots,n-1 \}.$
 \item $J_{V_k}(\xi) \subset J_V(\xi)$, for a.e. $\xi \in [0,1]$ and $k \in \{0, \ldots,n-1\}.$
 \end{enumerate}
 Moreover, in case these hold, we have  
 \begin{equation}\label{101}
V= V_0 \dot{\oplus} \cdots \dot{\oplus} V_{n-1}
\end{equation}
with each $V_k$ being a (possibly trivial) $\frac{\mathbb{Z}}{n}$-extra invariant FSIS.
\end{theorem}

\begin{definition}
For a given set of vectors $V=\{f_1, \ldots, f_m \}$ in a Hilbert space $\mathcal{H}$, we define $\mathcal{B}(V)$ as the matrix 
\begin{equation} \label{301}
    \left[  \mathcal{B}(V) \right]_{i,j}= \langle f_i, f_j \rangle_{\mathcal{H}}, \quad \forall i,j=1, \ldots, m.
\end{equation}    
\end{definition}     
\begin{theorem}\cite{aldroubiacha2007} \label{303}
    Let $\mathcal{H}$ be an infinite dimensional Hilbert space, $\mathcal{F}=\{f_1, \ldots, f_m\} \subset \mathcal{H}, \mathcal{X}= \Span \{f_1, \ldots, f_m \},$ $ \lambda_1 \geq \cdots \geq \lambda_m$ be the eigenvalues of the matrix $\mathcal{B}(\mathcal{F})$ (where $\mathcal{B}(\mathcal{F})$ is as defined in \eqref{301}) and $y_1, \ldots, y_m \in \mathbb{C}^m$, with $y_i=(y_{i1}, \ldots, y_{im})^t$ be the orthonormal left eigenvectors associated with the eigenvalues $\lambda_1, \ldots, \lambda_m$. 
 
Let $n\le m$ be a non-negative integer. Define the vectors $q_1, \ldots, q_n \in \mathcal{H}$ by 
 \begin{equation}
     q_i=\widetilde{\sigma_i} \sum_{j=1}^m y_{ij}f_j, \quad \forall\, i=1, \ldots,n,
 \end{equation}
 where $\widetilde{\sigma_i}= \lambda_i^{\frac{1}{2}}$ if $\lambda_i \neq 0
,$ and $\widetilde{\sigma_i}=0$ otherwise. Then $\{q_1, \ldots,q_n \}$ is a Parseval frame of $W= \Span \{q_1, \ldots, q_n \}$ and the subspace $W$ is optimal in the sense that  \[\sum_{i=1}^m \left\|f_i - P_W f_i \right\|^2 \leq \sum_{i=1}^m \left\| f_i - P_{W'} f_i \right\|^2, \quad  \forall \quad \text{subspaces } W', \Dim W' \leq n.\] \end{theorem}
% Fix $n \in \mathbb{N}$. Given $\Phi=\{\phi_1, \cdots, \phi_m \}$ a finite collection of functions in $L^2(\mathbb{R})$, the $n_0$-Grammian $G^{n_0}_\Phi$ of $\Phi$ is the $ m \times m$ matrix of $n_0\mathbb{Z}$-periodic functions 
% \begin{equation}
%     \left[G^{n_0}_\Phi(\xi) \right]_{ij}=\sum_{k \in \mathbb{Z}} \widehat{\phi_i}(\xi +k n_0)\overline{\widehat{\phi_j}(\xi+k n_0)}. 
% \end{equation}
% One important property of the $n_0$-Grammian is given by the following lemma concerning the measurability of the eigenvalues and the existence of measurable eigenvectors of a non-negative matrix with measurable entries.
\begin{lemma}\cite[Lemma 2.3.5]{roncanadian1995} \label{304}
    Let $G(\xi)$ be an $m \times m $ self-adjoint matrix of measurable functions defined on a measurable subset $ E \subset \mathbb{R}$ with eigenvalues $ \lambda_1(\xi) \geq \cdots \geq \lambda_m(\xi)$. Then the eigenvalues $ \lambda_i$, $ i =1, \ldots,m$, are measurable  on $E$ and there exists an $m\times m$ matrix of measurable functions $U=U(\xi)$ on $E$ such that $U(\xi)U^*(\xi)=I$ for a.e. $ \xi \in E$ and such that 
    \begin{equation}\label{309}
        G(\xi)=U(\xi) \Lambda
        (\xi) U^*(\xi), \quad
         \text{ for a.e. }  \xi \in E,
    \end{equation}
     where $\Lambda(\xi):=\Diag(\lambda_1(\xi), \cdots, \lambda_m(\xi))$.
\end{lemma}

% We will need definitions and known results concerning extra-invariance for shift-invariant spaces. These are described in this subsection.

% \begin{definition}
%     Let $V \subset L^2(\mathbb{R})$ be a SIS. We define the invariance set as follows
%     \[M:=\{ x \in \mathbb{R}: T_xf \in V, \forall f \in V \}.\]
% \end{definition}
% \begin{definition}
%     Let $\Phi \subset L^2(\mathbb{R})$. We will say that $V=V(\Phi)$ is $M$- extra invariant if $T_mf \in V$ for all $m \in M$ and for all $f \in V$. If $M=\mathbb{R}$, then the space $V$ is translation invariant but generated by the integer translations of the set $\Phi$.
% \end{definition}

%  \begin{theorem}\label{102}
%         If $V \subset L^2(\mathbb{R})$ is an SIS, then following are equivalent:
%         \begin{enumerate}
%             \item $V$ is $\frac{\mathbb{Z}}{n}$-invariant.
%             \item $V_k \subset V$ for $ k=0, \ldots, n-1$.
%             \item If $f \in V $, then $f^k=P_kf \in V$ for each $k=0, \ldots,n-1$.
%         \end{enumerate}
    
%     Moreover, in case these  hold, we have that $V$ is the orthogonal direct sum 
%     \begin{equation}\label{101}
%         V= V_0 \dot{\oplus}
%     \cdots\dot{\oplus} V_{n-1}
%     \end{equation}
%     with  each  $V_k$ being a (possibly trivial) $\frac{\mathbb{Z}}{n}$-invariant SIS.
% \end{theorem}

\section{Problem setup}\label{problemsetup}
\subsection{Statement of the minimization problem}
\text{ }\\

As mentioned in the introduction, our aim is to find the space closest to the given measurements. We now present the problem in a mathematically rigorous form. 
\smallskip

First, we define what we mean by the measurements of a function $f$ in $L^2(\mathbb{R})$. Fix $n_0 \in \mathbb{N}$. Let $g \in L^2(\mathbb{R})$ be such that there exists $ M > 0$ satisfying
\begin{align}
    \sum_{k \in \mathbb{Z}} \left| \widehat{g} \left( \xi + n_0k \right) \right|^2 \leq M \quad \textnormal{ for a.e.} \quad \xi \in [0,n_0].
    \label{1}
\end{align}

Define the sampling operator $S_g^{n_0} \colon L^2(\mathbb{R}) \rightarrow \ell^2(\mathbb{Z})$ by $S_g^{n_0}(f)=\left\lbrace \left\langle f,T_{\frac{k}{n_0}}g\right\rangle \right\rbrace_{k\in \mathbb{Z}}.$ 
Using~\eqref{1}, it can be easily verified that $S_g^{n_0}$ is a well-defined bounded linear operator. Motivated by the definition of $S^{n_0}_g$, we refer to $n_0$  as the sampling rate. For ease of notation, we denote $$f\left(\frac{k^g}{n_0}\right):=\left\langle f,T_{\frac{k}{n_0}}g\right\rangle,\quad \forall\, k\in \mathbb{Z}.$$ 
We assume that the sampled values/measurements $y=\left\lbrace y_k\right\rbrace_{k\in \mathbb{Z}}$  of a function $f\in L^2(\mathbb{R})$, taking into account the measurement error, have the following form:  For each $ k\in \mathbb{Z}$, 
\begin{align*}
     y_k= f\left(\frac{k^g}{n_0}\right) + n_k^f,
\end{align*}
where the error sequence $\left\lbrace n_k^f \right\rbrace_{k\in \mathbb{Z}} \in \ell^2(\mathbb{Z}).$
Note that \[\left(\sum_{k\in \mathbb{Z}} \left|  y_k\right|^2\right)^{\frac{1}{2}}=\norm{  S_g^{n_0}(f)+\left\lbrace n_k^f \right\rbrace_{k\in \mathbb{Z}}} \leq \norm{ S_g^{n_0}(f)} + \norm{\left\lbrace n_k^f \right\rbrace_{k\in \mathbb{Z}}}  < \infty.\] 
Hence, when we say that we are given measurements of the functions $\mathcal{F}=\left\lbrace  f_1,\ldots,f_m\right\rbrace$ in $L^2(\mathbb{R})$, we mean that the sequences
$\mathcal{F}_Y:=\left\lbrace  \left\lbrace  y_k^1 \right\rbrace_{k\in \mathbb{Z}},\ldots, \left\lbrace  y_k^m \right\rbrace_{k\in \mathbb{Z}}  \right\rbrace \subseteq \ell^2(\mathbb{Z})$ are given to us.

\smallskip
As stated in the introduction our \textbf{minimization problem} can be divided into two parts.
\begin{enumerate}
    \item  
    Fix $\lambda > 0$. Pick a $V \in \mathcal{C}$ and solve  initially:
    \begin{align*}
         % & \arg\min_{f\in V} \norm{ \left\lbrace y_k^j \right\rbrace_{k \in \mathbb{Z}} - S_g^{n_0}(f) }_{\ell^2(\mathbb{Z})}^2+\lambda \norm{f}^2\\
        % & \arg\min_{f\in V} \sum_{k\in \mathbb{Z}} \left| y_k^{j}-\left\langle f,g\left( -\frac{k}{n_0} \right)  \right\rangle\right|^2 + \lambda \norm{f}^2\\
        & \arg\min_{f\in V} \left\{\sum_{k\in \mathbb{Z}} \left| y_k^j - f\left( \frac{k^g}{n_0}\right)\right|^2 + \lambda \norm{f}^2\right\}
       \numberthis \label{3}.
    \end{align*}
    \item
   Subsequently, compute:
    \begin{equation} \label{4}
       \arg\min_{V \in \mathcal{C}}\sum_{j =1} ^m\left(\min_{f\in V} \sum_{k\in \mathbb{Z}} \left| y_k^j - f\left( \frac{k^g}{n_0}\right)\right|^2 + \lambda \norm{f}^2\right).
    \end{equation}
\end{enumerate}

Before we make a choice for $\mathcal{C}$ and proceed further, we show that the above minimization problem \eqref{4} can be restated in a simpler form using orthogonal projections (see \cite[Subsection 6.5.1]{pilonetto}).
\begin{definition}
\begin{enumerate}
\item Fix $\lambda>0$. Define the space
\begin{equation}
\mathcal{R}:=\{(c,f):c\in \ell^2(\mathbb{Z}), f \in L^2(\mathbb{R})\}.
\end{equation}
It forms a Hilbert space when endowed with the following inner product. For $(c_1,f_1),(c_2,f_2) \in \mathcal{R}$,
\[  \left\langle (c_1,f_1),(c_2,f_2) \right\rangle := \left\langle c_1,c_2 \right\rangle_{\ell^2(\mathbb{Z})} + \lambda   \left\langle f_1,f_2 \right\rangle_{L^2(\mathbb{R})}.\]
Let $\mathcal{R}_\lambda:=\left(\mathcal{R},\langle\cdot ,\cdot 
 \rangle\right)$. The subscript $\lambda$ is added to emphasize the fact that the inner product depends on $\lambda$.
 \item For any $f\in L^2(\mathbb{R})$,
 \[\widetilde{f}:=\left(S_g^{n_0}\left(f\right),f\right).\]
It can be verified that  $\hspace{0.1cm} \widetilde{ } \hspace{0.1cm}\colon L^2(\mathbb{R}) \to \mathcal{R_{\lambda}}$ mapping $f$ to $\widetilde{f}$ is one-one. Note that the $ \hspace{0.1cm} \widetilde{ } \hspace{0.1cm}$ operator is implicitly dependent on the measurement rate $n_0$.
 \item For any closed subspace $V$ of  $L^2(\mathbb{R})$,
 \[ \widetilde{V}:=\left\{\left(S_g^{n_0}(f),f\right):f\in V \right\}.\]
 Again, it is easy to verify that the map $\widetilde{ }\colon$ the collection of closed subspaces of $L^2(\mathbb{R}) \to $ the collection of closed subspaces of $\mathcal{R_{\lambda}}$ mapping $V$ into $\widetilde{V}$ is a well-defined one-one map.
  \end{enumerate}
\end{definition}

\begin{remark}
    Let $V\subseteq L^2(\mathbb{R})$ be a closed subspace. As a consequence of the above statements, any element in $\widetilde{V}$ will be represented by $\widetilde{f}$ for some $f\in V$.
\end{remark}

Fix $V \in \mathcal{C}$. Let $P_{\widetilde{V}}:\mathcal{R_{\lambda}} \to \widetilde{V}$ denote the orthogonal projection of $\mathcal{R_{\lambda}}$ onto $\widetilde{V}$. Define
\begin{equation}\label{5}
    Y^j=\left(\left\{y_k^j\right\}_{k\in \mathbb{Z}},\boldsymbol{0}\right),\; \forall j\in \{1,\dots, m\}.
\end{equation}
Here $\boldsymbol{0}$ denotes the zero vector in $L^2(\mathbb{R})$. Clearly, $Y_j \in \mathcal{R_{\lambda}}$, $\forall j\in \{1,\dots,m\}$.
Further, by definition
\begin{align*}
    \left\|Y^j-P_{\widetilde{V}} Y^j\right\|_{\mathcal{R_{\lambda}}}&=\min_{\widetilde{f}\in \widetilde{V}} \left\| Y^j-\widetilde{f}\right\|_{\mathcal{R_{\lambda}}},
\end{align*}
which implies that
\begin{align*} 
 \left\|Y^j-P_{\widetilde{V}} Y^j\right\|_{\mathcal{R_{\lambda}}}^2 &= \min_{\widetilde{f}\in \widetilde{V}} \left\| Y^j-\widetilde{f}\right\|_{\mathcal{R_{\lambda}}}^2 \\
   &= \min_{\widetilde{f}\in \widetilde{V}} \left \| \left(\{y_k^j\}_{k\in \mathbb{Z}},\boldsymbol{0}\right) -(S_g^{n_0}(f),f) \right \|^2 \\
   &= \min_{\widetilde{f}\in \widetilde{V}} \left\| (\{y_k^j\}_{k\in \mathbb{Z}} -S_g^{n_0}(f),-f)\right\|^2\\
    &= \min_{\widetilde{f}\in \widetilde{V}} \sum_{k\in \mathbb{Z}}\left|y_k^j-f\left(\frac{k^g}{n_0}\right)\right|^2+ \lambda \|f\|^2 \\
 &= \min_{f\in V} \sum_{k\in \mathbb{Z}}\left|y_k^j-f\left(\frac{k^g}{n_0}\right)\right|^2+ \lambda \|f\|^2\numberthis\label{6}.
\end{align*}
That is,
\begin{equation}\label{7}
    P_{\widetilde{V}}Y^j= \arg\min_{\widetilde{f}\in \widetilde{V}} \sum_{k \in \mathbb{Z}} \left|y_k^j-f\left(\frac{k^g}{n_0}\right)\right|^2 + \lambda \|f\|^2.
\end{equation}
Using \eqref{6} and \eqref{7}, we can conclude two things. Firstly, the minimizer $f^{j,\#}_V$ of \eqref{3} exists and satisfies $\widetilde{f^{j,\#}_V}=P_{\widetilde{V}}Y^j$. Secondly, our minimization problem \eqref{4} can be written as
\begin{equation}\label{8}
    \hspace{-5cm}\textbf{Minimization Problem Form 1:}\quad \quad \quad \arg\min_{V\in \mathcal{C}} \sum_{j=1}^m\left\|Y^j-P_{\widetilde{V}}Y^j\right\|^2_\mathcal{R_{\lambda}}.
\end{equation}
The benefit of rewriting our minimization in the above way is that it now aligns with the form considered in \cite{aldroubiacha2007}, allowing us to use the techniques present in the existing literature.

The next step is to further reduce the minimization problem. Note that, as $V\subseteq L^2(\mathbb{R})$, we have
\[ \widetilde{V}\subseteq \widetilde{L^2(\mathbb{R})} \subseteq \mathcal{R_{\lambda}} .\]
Now, consider
\begin{align*}
    \sum_{j=1}^m \left\|Y^j- P_{\widetilde{V}}Y^j \right\|_\mathcal{R_{\lambda}}^2 &= \sum_{j=1}^m \left\|Y^j-P_{\widetilde{L^2(\mathbb{R})}}Y^j+ P_{\widetilde{L^2(\mathbb{R})}}Y^j-   P_{\widetilde{V}}Y^j\right\|_\mathcal{R_{\lambda}}^2\\
    \numberthis \label{9}&= \sum_{j=1}^m\left(  \left\| Y^j-P_{\widetilde{L^2(\mathbb{R})}}Y^j \right\|^2  + \|P_{\widetilde{L^2(\mathbb{R})}}Y^j-   P_{\widetilde{V}}Y^j\|_\mathcal{R_{\lambda}}^2 \right).
\end{align*}
Indeed, as
\begin{align*}
   \left\langle Y^j-P_{\widetilde{L^2(\mathbb{R})}}Y^j, P_{\widetilde{L^2(\mathbb{R})}}Y^j-P_{\widetilde{V}}Y^j  \right\rangle  &=   \left\langle Y^j, P_{\widetilde{L^2(\mathbb{R})}}Y^j  \right\rangle -  \left\langle Y^j, P_{\widetilde{V}}Y^j  \right\rangle \\ & \quad -  \left\langle P_{\widetilde{L^2(\mathbb{R})}}Y^j, P_{\widetilde{L^2(\mathbb{R})}}Y^j \right\rangle +  \left\langle P_{\widetilde{L^2(\mathbb{R})}}Y^j, P_{\widetilde{V}}Y^j  \right\rangle \\
   &=   \left\langle Y^j, P_{\widetilde{L^2(\mathbb{R})}}Y^j  \right\rangle  -  \left\langle Y^j, P_{\widetilde{V}}Y^j  \right\rangle \\& \quad-
   \left\langle Y^j, P_{\widetilde{L^2(\mathbb{R})}}^* P_{\widetilde{L^2(\mathbb{R})}}Y^j \right\rangle +\left\langle  Y^j,  P_{\widetilde{L^2(\mathbb{R})}}^* P_{\widetilde{V}}Y^j  \right\rangle\\
   &=0.
\end{align*}
In the last statement we have used the fact that  $P_{\widetilde{L^2(\mathbb{R})}}^* P_{\widetilde{L^2(\mathbb{R})}}=P_{\widetilde{L^2(\mathbb{R})}}$ and $P_{\widetilde{L^2(\mathbb{R})}}^* P_{\widetilde{V}}=P_{\widetilde{V}}$. Furthermore, $P_{\widetilde{V}} P_{\widetilde{L^2(\mathbb{R})}}= P_{\widetilde{V}} \left(  P_{\widetilde{V}}+ P_{\widetilde{V}^{\perp}_{  \widetilde{L^2(\mathbb{R})} }} \right)=P_{\widetilde{V}}$.
Here, ${\widetilde{V}^{\perp}_{  \widetilde{L^2(\mathbb{R})} }}$ denotes the orthogonal complement of $\widetilde{V}$ considered as a subspace of $\widetilde{L^2(\mathbb{R})}$. Hence 
we can conclude from \eqref{9} that
 \begin{align*}
     \arg\min_{V\in \mathcal{C}} \sum_{j=1}^m \left\|Y^j-P_{\widetilde{V}}Y^j\right\|^2_\mathcal{R_{\lambda}}
     &= \arg\min_{V\in \mathcal{C}} \sum_{j=1}^m \left\|Y^j-P_{\widetilde{L^2(\mathbb{R})}}Y^j\right\|^2  +  \sum_{j=1}^m \left\|P_{\widetilde{L^2(\mathbb{R})}}Y^j -P_{\widetilde{V}}P_{\widetilde{L^2(\mathbb{R})}}Y^j\right\|^2\\
     &= \arg\min_{V\in \mathcal{C}} \sum_{j=1}^m  \left\|P_{\widetilde{L^2(\mathbb{R})}}Y^j -P_{\widetilde{V}}P_{\widetilde{L^2(\mathbb{R})}}Y^j\right\|^2.
 \end{align*}

\begin{definition}
      For each $j\in \{1,2,\dots,m\},$ define $f_{Y,j}$ as the function in $L^2(\mathbb{R})$ satisfying
  \begin{equation}\label{10}
      \widetilde{f_{Y,j}}=P_{\widetilde{L^2(\mathbb{R})}}Y^j \hspace{0.2cm }.
  \end{equation}
\end{definition}
Thus our minimization problem \eqref{4} can also be written as
\begin{equation}\label{minprobform2}
    \hspace{-5cm}\textbf{Minimization Problem Form 2:}\quad \quad \quad \arg\min_{V\in \mathcal{C}} \sum_{j=1}^m\left\|\widetilde{f_{Y,j}}-P_{\widetilde{V}}\widetilde{f_{Y,j}}\right\|^2_\mathcal{R_{\lambda}}.
\end{equation}

As mentioned earlier, the goal is to restate the minimization problem using the developed tools to ultimately arrive at a form that can be solved with the techniques available to us. For this reason, we explicitly calculate $\left \{\widetilde{f_{Y,j}}\right\}_{j=1}^m$ as it helps us reach a more solvable form. For this, we first introduce our fiber map.\\
\subsection{Fiber Map}
% \begin{enumerate}
%     \item Let the space $\mathbb{C} \times_\lambda \ell^2(\mathbb{Z}):=\left\{ (\alpha,a): \alpha \in \mathbb{C}, a \in \ell^2(\mathbb{Z}) \right\})$ with the following inner product. For $(\alpha,a), (\beta,b) \in \mathbb{C} \times_\lambda L^2(\mathbb{Z})$,
% \[\left\langle (\alpha,a), (\beta,b) \right\rangle= \langle \alpha, \beta \rangle + \lambda \langle a, b \rangle.\]
%     \item And, \[L^2 \left( [0,n_0], \mathbb{C} \times_\lambda \ell^2(\mathbb{Z}) \right):= \left\{ \Phi: [0,n_0] \rightarrow \mathbb{C} \times_\lambda \ell^2(\mathbb{Z}) : f \text{ is measurable and } \int_0^{n_0} \left\| \Phi(\xi) \right\|^2_{\mathbb{C} \times_\lambda \ell^2(\mathbb{Z})} d\xi < \infty \right\}\] with the following inner product. For $\Phi, \Psi \in L^2 \left( [0,n_0], \mathbb{C} \times_\lambda \ell^2(\mathbb{Z}) \right),$ \[\left\langle \Phi, \Psi \right\rangle= \int_0^{n_0} \left\langle\Phi(\xi), \Psi(\xi) \right\rangle_{\mathbb{C} \times_\lambda \ell^2(\mathbb{Z})} d\xi. \]

% \end{enumerate}
%      Under the above defined inner products, both $\mathbb{C} \times_\lambda \ell^2(\mathbb{Z})$ and $L^2 \left( [0,n_0], \mathbb{C} \times_\lambda \ell^2(\mathbb{Z}) \right)$ form Hilbert spaces. 

In this subsection, we define  the generalized version of the classical fiber map. Unlike the classical map, which is defined on functions in $L^2(\mathbb{R})$, this new map is defined on vectors $(c,f) \in \mathcal{R}_\lambda$. Specifically, when $(c,f) = \left(\left\{f(\frac{k^g}{n_0})\right\}_{k \in \mathbb{Z}},f\right)$ for some $f \in L^2(\mathbb{R})$, this map considers both the uniform measurements and the function together. Defining this new fiber map is essential because, as in the traditional case, it is necessary to transition to the Fourier domain to fully leverage the structural properties of FSISs. Since the sampling rate is $n_0$, the fiber map is defined on $[0,n_0]$ to ensure that the critical property \eqref{201} can be established.

First, we define the following two spaces.
\begin{definition}
\begin{enumerate}
\item Define the space 
\[\mathbb{C} \times \ell^2(\mathbb{Z}):= \{ (\alpha, a): \alpha \in \mathbb{C}, a \in \ell^2(\mathbb{Z}) \}.\]
It forms a Hilbert space when endowed with the following inner product. For $(\alpha,a), (\beta,b) \in \mathbb{C} \times \ell^2(\mathbb{Z})$,
\[ \langle (\alpha,a), (\beta,b) \rangle = \langle \alpha,\beta \rangle + \lambda \langle a, b \rangle .\]
Let $ \mathbb{C} \times_\lambda \ell^2(\mathbb{Z}):= (\mathbb{C} \times \ell^2(\mathbb{Z}), \langle \cdot, \cdot \rangle)$. Again, like in the case of $\mathcal{R}_\lambda$, the subscript is added to emphasize the dependence of the inner product on $\lambda$.\\
\item Define the space 
\begin{align*}
\hspace{2cm} L^2 \left( [0,n_0], \mathbb{C} \times_\lambda \ell^2(\mathbb{Z}) \right):=&\\ &\hspace{-2cm}\left\{ \Phi: [0,n_0] \rightarrow \mathbb{C} \times_\lambda \ell^2(\mathbb{Z}) : \phi \text{ is measurable and } \int_0^{n_0} \left\| \Phi(\xi) \right\|^2_{\mathbb{C} \times_\lambda \ell^2(\mathbb{Z})} d\xi < \infty \right\}.    
\end{align*} 
\hspace{-1cm} It forms a Hilbert space when endowed with the following inner product. For $\Phi, \Psi \in L^2 \left( [0,n_0], \mathbb{C} \times_\lambda \ell^2(\mathbb{Z}) \right),$ \[\left\langle \Phi, \Psi \right\rangle= \int_0^{n_0} \left\langle\Phi(\xi), \Psi(\xi) \right\rangle_{\mathbb{C} \times_\lambda \ell^2(\mathbb{Z})} d\xi. \]
\end{enumerate}
\end{definition}

 \begin{lemma}\label{gammaisisom}
    The fiber map $\widetilde{\Gamma}:\mathcal{R_{\lambda}} \rightarrow L^2([0,n_0], \mathbb{C}\times_\lambda \ell^2(\mathbb{Z}))$ defined for all $(c,f) \in \mathcal{R_{\lambda}}$ as
    \begin{equation}\label{36}
        (\widetilde{\Gamma}(c,f))(\xi)=\left( \sum_{k \in \mathbb{Z}}c_k e^{-\frac{2 \pi i k \xi }{n_0}}, \left\{\widehat{f}(\xi+ kn_0)\right\}_{k \in \mathbb{Z}}\right), \text{ for a.e. } \xi \in [0,n_0],
    \end{equation}
    is an isometric isomorphism
\end{lemma}
\begin{proof}
    It is straightforward.
\end{proof}

\smallskip
Restricted to $\widetilde{L^2(\mathbb{R})}$, the fiber map $\widetilde{\Gamma}$ has a specific form. For any $\widetilde{f} \in \widetilde{L^2(\mathbb{R})}$, 
    \begin{align}
    \label{15}
        \left( \widetilde{\Gamma} \widetilde{f}\right) (\xi) = \left( \sum_{k \in \mathbb{Z}}  f \left( \frac{k^g}{n_0} \right) e^{-\frac{2\pi i k \xi}{n_0}}, \left\lbrace  \widetilde{f} \left( \xi + k n_0\right) \right\rbrace_{k \in \mathbb{Z}}\right), \text{ for a.e. } \xi \in [0,n_0].
    \end{align}
    Further, note that $\text{for a.e. } \xi \in [0,n_0]$,
\begin{align}
    \sum_{k \in \mathbb{Z}} f \left( \frac{k^g}{n_0}\right) e^{-\frac{2\pi i k \xi}{n_0}} &= \sum_{k \in \mathbb{Z}} \left\langle f, g \left( \cdot- \frac{k}{n_0} \right)  \right\rangle e^{-\frac{2\pi i k \xi}{n_0}} =\sum_{k \in \mathbb{Z}} \left\langle \widehat{f}, e^{-\frac{2\pi i k \cdot }{n_0}} \widehat{g} \right\rangle e^{-\frac{2\pi i k \xi}{n_0}} \nonumber
    \\
    & = \sum_{k \in \mathbb{Z}} \left( \int_{\mathbb{R}} \widehat{f} (\eta) \overline{\widehat{g} (\eta)} e^{\frac{2\pi i k \eta }{n_0}} d\eta   \right) e^{-\frac{2\pi i k \xi }{n_0}}
    \nonumber
    \\
    \label{16}
    &=\sum_{k\in \mathbb{Z}} \left( \int_0^{n_0} \sum_{l \in \mathbb{Z}} \widehat{f} (\eta + l n_0) \overline{\widehat{g}(\eta+l n_0)} e^{\frac{2\pi i k \eta}{n_0}} d\eta\right) e^{-\frac{2\pi i k \xi}{n_0}}
    \\
    & = \sum_{k\in \mathbb{Z}} \left\langle \sum_{l \in \mathbb{Z}} \widehat{f} (\cdot + l n_0) \overline{\widehat{g}(\cdot+l n_0)}, e^{-\frac{2\pi i k \cdot}{n_0}}  \right\rangle_{L^2[0,n_0]} e^{-\frac{2\pi i k \xi}{n_0}}
    \nonumber
    \\
    & = \sum_{l \in \mathbb{Z}} \widehat{f} (\xi + l n_0) \overline{\widehat{g}(\xi+l n_0)}\nonumber.
\end{align}
The equality \eqref{16} was obtained using \cite[Lemma 9.2.3]{christensen2016book}.
Hence, we can conclude that for all $\widetilde{f} \in \widetilde{L^2(\mathbb{R})}$, 
\begin{align}
    \label{18}
    \left(  \widetilde{\Gamma} \widetilde{f} \right)(\xi) = \left( \sum_{l \in \mathbb{Z}} \widehat{f} (\xi + l n_0) \overline{\widehat{g}(\xi+l n_0)}, \left\lbrace  \widehat{f} (\xi + ln_0) \right\rbrace_{l \in \mathbb{Z}}  \right), \hspace{0.2cm} \text{for a.e. }  \xi \in [0,n_0].
\end{align}
\begin{remark}
The equality \eqref{18} essentially states that for any function $f \in L^2(\mathbb{R})$ and for a.e. $\xi \in [0,n_0]$, $\left(\Tilde{\Gamma}\Tilde{f}\right)(\xi)$ is a vector consisting of the sequence $\left\lbrace \widehat{f} (\xi + ln_0) \right\rbrace_{l \in \mathbb{Z}}$ along with a linear combination of itself with coefficients as  $\left\lbrace \overline{\widehat{g} (\xi + ln_0)} \right\rbrace_{l \in \mathbb{Z}}$. This form is especially useful and will be used later to prove Lemma \ref{decompositionoftilde{V}}.
\end{remark}
\subsection{Calculation of $f_{Y,j}$}\label{fYj}
\text{ }

Having defined our fibre map, we now refocus on our aim of  calculating $\{f_{Y,j}\}_{j\in \{0,\ldots,n_0-1\}}$. Fix $j\in \{0,\ldots,n_0-1\}$, then from \eqref{10}, we have 
\begin{align*}
    \widetilde{f_{Y,j}}&=\arg\min_{\widetilde{f} \in \widetilde{L^2(\mathbb{R})}} \left\|Y^j-\widetilde{f} \right\|^2_{\mathcal{R}_\lambda}=\arg\min_{\widetilde{f} \in \widetilde{L^2(\mathbb{R})}} \left\|\widetilde{\Gamma}Y^j- \widetilde{\Gamma}\widetilde{f} \right\|^2_{L^2([0,n_0], \mathbb{C}\times_\lambda \ell^2(\mathbb{Z}))} \\
    &=\arg\min_{\widetilde{f} \in \widetilde{L^2(\mathbb{R})}} \int_0^{n_0} \left\| \left( \widetilde{\Gamma} Y^j \right) (\xi) -\left( \widetilde{\Gamma} \widetilde{f} \right)(\xi) \right\|^2_{\mathbb{C}\times_\lambda \ell^2(\mathbb{Z}))} d\xi \\
\numberthis \label{41}&=\arg\min_{\widetilde{f}\in \widetilde{L^2(R)}}\int_0^{n_0}\left(\left|\sum_{l\in \Z } y_l^je^{-\frac{2\pi il\xi}{n_0}}-\sum_{l\in \Z }\widehat{f}(\xi+ln_0)\overline{\widehat{g}(\xi+ln_0)} \right|^2+\lambda \sum_{l\in \Z}\left|\widehat{f}(\xi+ln_0)\right|^2\right)d\xi.
\end{align*}
For a.e. $\xi\in [0,n_0]$, define the space \begin{equation}\label{42}
    A_{\xi}:=\left\{\left(\sum_{l\in \Z} a_l\overline{\widehat{g}\left(\xi+ln_0\right)},a   \right)\colon a\in \ell^2(\Z)   \right\}.
\end{equation} 
It forms a closed subspace of $\mathbb{C}\times_{\lambda}\ell^2(\Z)$ for a.e. $\xi \in [0,n_0]$. 

\smallskip
 Given any $\widetilde{f} \in \widetilde{L^2(\mathbb{R})}$, $\left( \sum_{l \in \mathbb{Z}} \widehat{f}(\xi +l n_0) \overline{\widehat{g}(\xi+ l n_0)}, \left\{ \widehat{f}(\xi+ l n_0) \right\}_{l \in \mathbb{Z}} \right) \in A_\xi$ for a.e. $ \xi \in [0,n_0]$. Further, for a.e. $\xi \in [0,n_0]$, the term inside the integral in \eqref{41}, satisfies 
 \begin{align*}
     & \hspace{-2cm} \left| \sum_{l \in \mathbb{Z}} y^j_k e^{- \frac{2 \pi i l \xi}{n_0}} -\sum_{l \in \mathbb{Z}}\widehat{f}(\xi+ l n_0) \overline{\widehat{g}(\xi+l n_0)}\right|^2 + \lambda\sum_{l \in \mathbb{Z}} \left| \widehat{f}(\xi +l n_0) \right|^2\\
     &= \left\| \left( \sum_{l \in \mathbb{Z}} y^j_l e^{- \frac{2 \pi l \xi}{n_0}}, \boldsymbol{0} \right) -\left( \sum_{l \in \mathbb{Z}} \widehat{f}(\xi +l n_0) \overline{\widehat{g}(\xi +l n_0)}, \left\{ \widehat{f}(\xi + l n_0) \right\}_{ l \in \mathbb{Z}} \right) \right\|^2\\
     & \geq\left\| \left( \sum_{ l \in \mathbb{Z}} y^j_l e^{- \frac{2 \pi i l \xi}{n_0}}, \boldsymbol{0} \right) - P_{A_\xi} \left( \sum_{l \in \mathbb{Z}} y^j_l e^{- \frac{2 \pi i l \xi}{n_0}}, \boldsymbol{0} \right) \right\|^2.
 \end{align*}
 Suppose there exists $f_j \in L^2(\mathbb{R})$ such that $\left( \widetilde{\Gamma} \widetilde{f_j} \right) (\xi)=P_{A_\xi} \left( \sum_{l \in \mathbb{Z}} y^j_l e^{- \frac{2 \pi i l \xi}{n_0}}, \boldsymbol{0} \right)$ for a.e. $\xi \in [0,n_0]$, then  
 \begin{align*}
     & \hspace{-2cm} \left| \sum_{l \in \mathbb{Z}} \left( y^j_k e^{- \frac{2 \pi i l \xi}{n_0}} -\widehat{f}(\xi+ l n_0) \overline{\widehat{g}(\xi+l n_0)} \right) \right|^2 + \lambda\sum_{l \in \mathbb{Z}} \left| \widehat{f}(\xi +l n_0) \right|^2\\
      &\geq \left\| \left( \widetilde\Gamma\ Y_j \right)(\xi) -\left( \widetilde{\Gamma} \widetilde{f_j} \right) (\xi) \right\|^2. 
 \end{align*}
 Hence, it follows that
\[\int_0^{n_0} \left\|  \left(\widetilde{\Gamma} Y^j \right)(\xi) - \left( \widetilde{\Gamma} \widetilde{f} \right)(\xi) \right\|^2 d \xi \geq \int_{0}^{n_0} \left\|  \left( \widetilde{\Gamma} Y^j \right)(\xi) - \left( \widetilde{\Gamma} \widetilde{f_j} \right)(\xi) \right\|^2 d \xi,\]
which along with \eqref{41} implies that $ \widetilde{f_{Y,j}} = \widetilde{f_j}$.

\smallskip
So now we try to compute $P_{A_\xi} \left( \sum_{l \in \mathbb{Z}} y^j_l e^{- \frac{2 \pi i l \xi}{n_0}}, \boldsymbol{0} \right)$ explicitly for a.e. $\xi \in [0,n_0]$.

\smallskip
 Let $g_{\xi}$ denote the sequence $ g_{\xi}:=\left\{\widehat{g}\left( \xi+ln_0 \right)\right\}_{l\in \Z}$, for a.e. $\xi\in [0,n_0]$.
Then $A_{\xi}$ can be concisely written as $A_{\xi}=\left\{\left(\langle a,g_{\xi}\rangle,a\right)\colon a\in \ell^2(\Z)\right\}$.
Further, for a.e. $\xi \in [0,n_0]$, define $$B_{\xi}:=\left\{\left(\overline{\widehat{g}\left(\xi+ln_0\right)},e_l\right)\colon l\in \Z\right\}=\left\{\left(\langle e_l,g_{\xi}\rangle,e_l\right)\colon l\in \Z\right\}.$$
We claim that $B_\xi$ forms a Riesz basis for $A_\xi$, for a.e. $\xi \in [0,n_0]$.  It is easy to check that $B_\xi$ is complete in $A_\xi$, for a.e. $\xi \in [0,n_0]$. Let $b=\{b_l\}_{l \in \mathbb{Z}} \in \ell^2(\mathbb{Z})$ be a finite scalar sequence  Then, 
\begin{align*}
    \left\| \sum_{l \in \mathbb{Z}} b_l \left( \left\langle e_l, g_\xi \right\rangle , e_l \right) \right\|^2_{ \mathbb{C} \times_\lambda \ell^2(\mathbb{Z})} = \left\| \left( \sum_{l \in \mathbb{Z}} b_l \left\langle e_l ,g_\xi \right\rangle , \sum_{l \in \mathbb{Z}} b_l e_l \right) \right\|^2
    = \left| \sum_{l \in \mathbb{Z}} b_l \left\langle e_l, g_\xi \right \rangle \right|^2+ \lambda \left\| \sum_{l \in \mathbb{Z}} b_le_l \right\|^2
     \geq \lambda \|b\|^2.   
\end{align*}
Further, using \eqref{1},
\begin{align*}
    \left\| \sum_{l \in \mathbb{Z}} b_l \left( \left\langle e_l, g_\xi \right\rangle ,e_l \right) \right\|^2_{ \mathbb{C} \times_\lambda \ell^2(\mathbb{Z})}& \leq \left| \sum_{l \in \mathbb{Z}} b_l \left \langle e_l, g_\xi \right\rangle \right|^2+ \lambda \sum_{l \in \mathbb{Z}} |b_l|^2
     \leq  \sum_{l \in \mathbb{Z}} |b_l|^2\sum_{l \in \mathbb{Z}}\left| \left\langle e_l, g_\xi \right\rangle \right|^2+ \lambda\sum_{l \in \mathbb{Z}} |b_l|^2
     \leq (M+ \lambda )\|b\|^2,
\end{align*}
thereby, proving our claim.\\

We rearrange the basis $B_{\xi}$ as $$B_{\xi}=\left\{\left(\langle e_0,g_{\xi}\rangle,e_0\right),\left(\langle e_1,g_{\xi}\rangle,e_1\right),\left(\langle e_{-1},g_{\xi}\rangle,e_{-1}\right),\dots \right\}.$$
That is, $B_{\xi}=\left\{ \left(\langle \widetilde{e_n},g_{\xi}\rangle,\widetilde{e_n}  \right) \right\}_{n=0}^\infty$, where $\widetilde{e_0}=e_{0}$, $\widetilde{e_n}=e_{-\frac{n}{2}}$ if $n$ is even and $\widetilde{e_n}=e_{\frac{n+1}{2}}$ if $n$ is odd. The rearrangement can be done, as Riesz bases are unconditional. The next step is to orthonormalize the Riesz basis $B_{\xi}$ so that the orthogonal projection of any vector in $\mathbb{C}\times_{\lambda}\ell^2(\Z)$ onto $A_{\xi}$ can be computed. \\

Clearly, for each $n\geq 0$, the map from $[0,n_0]$ to $\mathbb{C}$ defined as, $\xi \mapsto a_n^{\xi}:=\langle \widetilde{e_n},g_{\xi}\rangle$ is measurable. Indeed, if $n$ is even, then $a_n^{\xi}=\overline{\widehat{g}\left(\xi-\frac{n}{2}n_0\right)}$ and if n is odd, then $a_n^\xi=\overline{\widehat{g}\left(\xi-\left(\frac{n+1}{2}\right)n_0\right)}$, both of which are measurable functions. From now on, we will use $a_n$ to denote $a_n^\xi$. However, note that $a_n$ is always implicitly dependent on $\xi$. For a.e. $\xi \in [0,n_0]$, we orthonormalize $B_\xi$ using the Gram-Schmidt orthogonalization process. The  orthonormalized basis $\left\{\frac{v^\xi_n}{\left\|v^\xi_n\right\|}\right\}_{n\geq 0}=\left\{\frac{v_n}{\|v_n\|}\right\}_{n\geq 0} $ can be  computed as follows.
\begin{align*}
    v_0&=\left(a_0,\widetilde{e_0}\right); \|v_0\|^2=|a_0|^2+\lambda,\\
    v_1&=\left(a_1,\widetilde{e_1}\right)-\left\langle\left(a_1,\widetilde{e_1}\right),\frac{v_0}{\|v_0\|}\right\rangle \frac{v_0}{\|v_0\|}\\
    &=\left(\frac{\lambda a_1}{|a_0|^2+\lambda},\Tilde{e_1}-\frac{a_1\overline{a_0}}{|a_0|^2+\lambda}\right);\|v_1\|^2=\frac{\lambda\left(|a_1|^2+|a_0|^2+\lambda\right)}{|a_0|^2+\lambda},\\
    \vdots\\
    \numberthis \label{vn} v_n&=\left(\frac{\lambda a_n}{|a_{n-1}|^2+\dots+|a_0|^2+\lambda},\Tilde{e_n}-\frac{a_n\overline{a_{n-1}}\widetilde{e_{n-1}}}{|a_{n-1}|^2+\dots+|a_0|^2+\lambda}\dots -\frac{a_n\overline{a_0}\widetilde{e_0}}{|a_{n-1}|^2+\dots+|a_0|^2+\lambda}\right);\\ &\|v_n\|^2=\frac{\lambda\left(|a_n|^2+\dots+|a_0|^2+\lambda\right)}{|a_{n-1}|^2+\dots+|a_0|^2+\lambda}.
\end{align*}
For $\xi\in [0,n_0]$, let $\left(a^\xi,\boldsymbol{0}\right)\in \mathbb{C}\times_\lambda \ell^2(\Z).$ Then 
\begin{align*}
    P_{A_\xi}\left(a^\xi,\boldsymbol{0}\right)&=\sum_{n=0}^\infty \left\langle \left(a^\xi,\boldsymbol{0}\right),\frac{v_n}{\|v_n\|}\right\rangle\frac{v_n}{\|v_n\|}=\sum_{n=0}^\infty \left\langle \left(a^\xi,\boldsymbol{0}\right),v_n\right\rangle\frac{v_n}{\|v_n\|^2}\\
    &=\sum_{n=0}^\infty \frac{a^\xi a_n}{\left(|a_n|^2+\dots+|a_0|^2+\lambda\right)} v_n=a^\xi \sum_{n=0}^\infty \frac{ a_n}{\left(|a_n|^2+\dots+|a_0|^2+\lambda\right)} v_n
\end{align*}
Choosing $\left(a^\xi,\boldsymbol{0}\right)=(1,0)$, we get
\begin{align*}
    P_{A_\xi}\left(1,\boldsymbol{0}\right)&= \sum_{n=0}^\infty \frac{ a_n}{\left(|a_n|^2+\dots+|a_0|^2+\lambda\right)} v_n
\end{align*}
As $P_{A_\xi}\left(1,\boldsymbol{0}\right) \in A_\xi$, there exists a $ d^\xi\in \ell^2(\Z)$ such that $$\left(\left\langle d^\xi,g_\xi\right\rangle,d^\xi\right)=\sum_{n=0}^\infty\frac{ a_n v_n^\xi}{\left(|a_n^\xi|^2+\dots+|a_0^\xi|^2+\lambda\right)}$$
Hence, we get  \begin{align*}
    \numberthis\label{302} P_{A_\xi}(a^\xi,0)&=a^\xi\left(\left\langle d^\xi,g_\xi\right\rangle,d^\xi\right)\\
    &=\left(\left\langle a^\xi d^\xi,g_\xi\right\rangle,a^\xi d^\xi\right).
\end{align*} 
As $a_n^\xi$ is measurable for each $n$, so is $v_n^\xi$, which in turn implies that $\xi\mapsto d^\xi$ is a measurable map on $[0,n_0].$ Next, in order to solve \eqref{41}, we make a particular choice for $a^\xi$.

Fix $j\in \{1,\dots,m\}$, and let $$a^\xi_j=\sum_{k\in \Z}y_k^j e^{-\frac{2\pi ik\xi}{n_0}}\quad \text{, for a.e. } \xi\in [0,n_0].$$ Then, $\xi\mapsto a^\xi d^\xi$ is a measurable map from $[0,n_0]$ to $\ell^2(\Z).$
In fact, we can show that it belongs to $L^2([0,n_0],\ell^2(\Z))$.  Consider
\begin{align*}
    \int_0^{n_0}\|a^\xi_j d^\xi\|^2_{\ell^2(\Z)}d\xi&\leq \left(\frac{1}{\lambda}\int_0^{n_0}|\langle a^\xi_j d^\xi,g_\xi\rangle|^2\right)+\int_0^{n_0}\|a^\xi_j d^\xi\|^2_{\ell^2(\Z)}d\xi\\
    &= \frac{1}{\lambda}\int_0^{n_0} \left(|\langle a^\xi_j d^\xi,g_\xi\rangle|^2+\lambda\|a^\xi_j d^\xi\|^2_{\ell^2(\Z)}\right)d\xi\\
    &=\frac{1}{\lambda}\int_0^{n_0}\|P_{A_\xi}(a^\xi_j,0)\|^2d\xi\\
& \leq \frac{1}{\lambda}\int_0^{n_0}\left\|\left(a^\xi_j,0\right)\right\|^2d\xi\\
    &=\frac{1}{\lambda}\int_0^{n_0}\left\vert \sum_{k\in \Z}y_k^je^{-\frac{2\pi ik\xi}{n_0}}\right\vert^2d\xi\\
    &=\frac{1}{\lambda}\sum_{k\in \Z}|y_k^j|^2 < \infty.
\end{align*}

Similarly, it can be shown that the map $\xi\mapsto d^\xi$
 also belongs to $L^2([0,n_0],\ell^2(\Z))$. Since the map $\Gamma\colon L^2(\mathbb{R})\rightarrow L^2([0,n_0],\ell^2(\Z))$, defined as $\Gamma f(\xi)=\left\{\widehat{f}(\xi+kn_0)\right\}_{k\in \Z}$ is an isometric isomorphism (for the case $n_0=1$, see \cite{bownikjfa2000}),  there exist unique $f$ and $f_j$ belonging to $L^2(\mathbb{R})$ satisfying $(\Gamma f)(\xi)=d^\xi$ and $(\Gamma f_j)(\xi)= a^\xi_j d^\xi$, for a.e. $\xi \in [0,n_0]$. Therefore, we can conclude that, for a.e. $\xi \in [0,n_0]$, 
\begin{align*}
   \numberthis \label{43} &\left( \widetilde{\Gamma} \widetilde{f} \right)(\xi)= \left( \left \langle d^\xi, g_\xi \right \rangle , d^\xi \right)  \text{ and } 
 \left( \widetilde{\Gamma} \widetilde{f_j} \right)(\xi)= \left( \left \langle a^\xi_jd^\xi, g_\xi \right \rangle , a^\xi_j d^\xi \right).
 \end{align*}
 That is,
 \begin{equation}\label{PAxi}
  P_{A_\xi} \left( \sum_{k \in \mathbb{Z}} y^j_k e^{- \frac{2 \pi i k \xi}{n_0}}, \boldsymbol{0} \right)= \left( \widetilde{\Gamma} \widetilde{f_j} \right)(\xi)= \sum_{k \in \mathbb{Z}}y^j_ke^{- \frac{2 \pi i k \xi}{n_0}} \left( \widetilde{\Gamma} \widetilde{f} \right)(\xi).
\end{equation}
 
%  there exists a unique $f\in L^2(\mathbb{R})$ such that $(\Gamma f)(\xi)=d^\xi $ a.e. $\xi\in [0,n_0].$ Therefore, we can conclude that 
% \begin{align*}
%      \left(\widetilde{\Gamma}\widetilde{f} \right)\left(\xi\right)&=\left( \langle d^\xi,g_\xi\rangle,d^\xi \right) \; \text{ for a.e.}\; \xi\in [0,n_0] 
%      \numberthis \label{43}\\
%   \text{i.e. } \quad   P_{A_\xi}\left(\sum_{k\in \Z}y_k^je^{-\frac{2\pi ik\xi}{n_0}}, \boldsymbol{0}\right)&=
%       \sum_{k\in \Z} y_k^je^{\frac{-2\pi ik\xi}{n_0}} (\widetilde{\Gamma} \widetilde{f})(\xi)\; \text{ for a.e. } \xi \in [0,n_0].
% \end{align*}
The last statement follows from \eqref{302}. 

\smallskip
Thus, finally, from \eqref{41},\eqref{PAxi} and from the computations we did at the beginning of this subsection, we can conclude that \begin{equation}\label{44}
     \left(\widetilde{\Gamma}\widetilde{f_{Y,j}}\right)(\xi)=\sum_{k\in \Z} y_k^j e^{-\frac{2\pi ik\xi}{n_0}}  \left(\widetilde{\Gamma}\widetilde{f}\right)(\xi), \text{ for a.e. } \xi\in [0,n_0].
 \end{equation}

\section{Fiber map theory  for FSISs}\label{fibermaptheory}

Now, we introduce some definitions and prove a few results related to our  fiber map. Since our fiber map is closely related to (and motivated by) the classical fiber map \cite{bownikjfa2000}, the theory is very similar. Therefore, proofs of the majority of our results are omitted, and proofs of results with significant changes are provided. 

\begin{definition}
    A range function $\widetilde{J}$ is a mapping $\widetilde{J}: [0,n_0] \rightarrow \left\{ \text{closed subspaces of } \mathbb{C} \times_\lambda \ell^2(\mathbb{Z}) \right\}.$
Given a range function $\widetilde{J}$, the space $M_{\widetilde{J}}$ is defined as 
\begin{equation}
    M_{\widetilde{J}}=\left\{ \Phi \in L^2([0,n_0], \mathbb{C}\times_\lambda \ell^2(\mathbb{Z})): \Phi (\xi) \in \widetilde{J}(\xi), \hspace{0.2cm} \text{for a.e. } \xi \in [0,n_0]\right\}.
\end{equation}
\end{definition}
\begin{remark}
    \begin{enumerate}
        \item  A range function $\widetilde{J}$ is called measurable if the associated orthogonal projections $P(\xi): \mathbb{C} \times_\lambda \ell^2(\mathbb{Z}) \rightarrow \widetilde{J}(\xi)$  are weakly operator measurable.
\item  Note that by the Pettis measurability theorem, the condition on $P$ is equivalent to the map $\xi \rightarrow P(\xi)a$ being vector measurable for each $a \in \mathbb{C} \times_\lambda \ell^2(\mathbb{Z})$. 
\item  Let $\widetilde{J}$ be a range function (need not be measurable). Then, it can be verified that $M_{\widetilde{J}}$ will form a closed subspace of $L^2([0,n_0], \mathbb{C}\times_\lambda \ell^2(\mathbb{Z})).$
\end{enumerate}
\end{remark}

\begin{lemma}
    Let $\widetilde{J}$ be a measurable range function with associated projections $P$.
    % Let $P_{M_{\widetilde{J}}}$ denote the orthogonal projection of $L^2\left([0,n_0], \mathbb{C} \times_{\lambda} \ell^2(\mathbb{Z})\right)$ onto $M_{\widetilde{J}}$.
    Then, for any $\Phi \in L^2([0,n_0], \mathbb{C}\times_\lambda \ell^2(\mathbb{Z})), $
 \begin{equation}
      \left(P_{M_{\widetilde{J}}} \Phi\right)(\xi) = P(\xi)\left(\Phi(\xi)\right),  \text{ for a.e.} \hspace{0.1cm} \xi \in [0,n_0].
 \end{equation}
\end{lemma}
% \begin{proof}
%     The proof being essentially the same as the classical case(\cite{bownikjfa2000}) is omitted. 
% \end{proof}
\begin{definition}
        For any $S \subseteq L^2(\mathbb{R})$, we define the range function $\widetilde{J}_{\widetilde{S}}$ as
 \begin{equation}\label{202}
     \widetilde{J}_{\widetilde{S}}(\xi):=\overline{\Span}\left\{\left(\widetilde{\Gamma}\widehat{\phi} \right)(\xi): \phi \in S \right\}, \text{ for a.e. } \hspace{0.1cm} \xi \in [0,n_0].
 \end{equation}

\end{definition}
 Further, for any $f \in L^2(\mathbb{R})$ and $k \in \mathbb{Z}$, it can be verified that 
\begin{equation}\label{201}
    \left(\widetilde{\Gamma}\widetilde{T_kf} \right)(\xi)=e^{-2 \pi i k \xi} \left(\widetilde{\Gamma}\widetilde{f} \right)(\xi),  \text{ for  a.e. }  \xi \in [0,n_0].
\end{equation} 
Using \eqref{202} and \eqref{201}, the following lemma can be shown
 
 \begin{lemma} \label{306}
     Let $\mathcal{{A}}=\left\{\phi_1, \ldots, \phi_l \right\}$ and $V=\overline{\Span}\left\{ \phi( \cdot - n): n \in \mathbb{Z}, \phi \in  \mathcal{A} \right\}$. Then
     \begin{equation}\label{JtildeVtilde=JtildeAtilde}
         \widetilde{J}_{\widetilde{V}}(\xi)= \widetilde{J}_{\widetilde{\mathcal{A}}}(\xi),  \text{ for a.e. } \xi \in [0,n_0].
     \end{equation}
 \end{lemma}
\begin{lemma}\label{Vtilde=spantilde}
Let $\phi_1, \ldots, \phi_l\in L^2(\mathbb{R})$ and $V=V(\phi_1, \ldots, \phi_l).$ Then
\[\widetilde{V}= \overline{\Span} \left\{ \widetilde{\phi_i(\cdot-k)}: i \in \{1,\ldots, l \}, k \in \mathbb{Z} \right\}.\]
\end{lemma}
\begin{proof}
 Note that, by definition $\widetilde{V}= \left\{ \widetilde{f}: f \in V \right\}= \left\{ \widetilde{f}: f \in \overline{\Span} \left\{ \phi_i(\cdot-k): i \in \{1, \ldots,l\}, k \in \mathbb{Z} \right\}\right\}$. Let $\widetilde{f} \in \widetilde{V}$. Then there exists a sequence $\{f_n\}_{n \in \mathbb{N}} \in \Span \{ \phi_i(\cdot-k): i \in \{1, \ldots,l\}, k \in \mathbb{Z} \}$ such that $f_n \rightarrow f$. Hence, using the fact that 
 \[\left\| \widetilde{f_n} - \widetilde{f} \right\|^2= \left\| \left( S^{n_0}g \left(f_n-f \right), f_n-f \right) \right\|^2 \leq  \left\| S^{n_0}_g \right\|^2 \left\|f_n -f \right\|^2 + \lambda \left\|f_n-f  \right\|^2,\]
we can conclude that $\widetilde{f_n} \rightarrow \widetilde{f}$.  Clearly $\widetilde{f_n} \in \overline{\Span} \left\{ \widetilde{\phi_i(\cdot-k)}: i \in\{1, \ldots, l \}, k \in \mathbb{Z} \right\} , $ for all $n \in \mathbb{N}$ . Therefore,  $\widetilde{f} \in \overline{\Span} \left\{ \widetilde{\phi_i(\cdot-k)}: i \in \{1, \ldots, l \}, k \in \mathbb{Z} \right\}$.\\

In order to prove the converse, let $\widetilde{f} \in \Span \left\{ \widetilde{\phi_i(\cdot-k)}: i \in \{1, \ldots ,l\}, k \in \mathbb{Z} \right\}$.  Then, for some positive integer $n$, $\widetilde{f}$ can be written as 
\begin{align*}
\widetilde{f}&= \alpha_1 \widetilde{\phi_{i_1}(\cdot- k_1)}+ \cdot + \alpha_n \widetilde{\phi_{i_n}(\cdot-k_n)}
=\widetilde{\sum_{j=1}^n \alpha_j \phi_{i_j}(\cdot-k_j)}
\end{align*}
Thus, $\widetilde{f} \in \widetilde{V}$.
\end{proof}
The following proposition is crucial to our theory . Let $k \in \{0, \ldots,n_0-1\}$. We prove that for any given FSIS $V \subset P_{k}(L^{2}(\mathbb{R}))$, the space $\Tilde{V}$ has an equivalent form that can be defined using our newly defined fiber map. The proof of this proposition relies primarily on two points. Firstly, on \eqref{201}, and secondly, on the fact that for any $f \in V$, 
the support of $\hat{f}$ is contained within $B_{k}$ (see \eqref{defintionofB_k}).
 \begin{prop}   \label{307}
     Fix $k \in \{0, \ldots,n_0-1\}$ and let $\mathcal{A}=\{\phi_1,\ldots,\phi_l\} \subseteq P_k(L^2(\mathbb{R}))$. Then
     \begin{enumerate}
         \item $V=\overline{\Span} \left\{ \phi(\cdot - n): n \in \mathbb{Z}, \phi \in \mathcal{A} \right\}$ if and only if 
         \begin{equation}\label{VtoVtilde}
             \widetilde{V}=\left\{ (c,f) \in \mathcal{R}_\lambda : \left( \widetilde{\Gamma} (c,f)\right) (\xi) \in \widetilde{J}_{\widetilde{\mathcal{A}}}(\xi)  \text{ for a.e. }  \xi \in [0,n_0] \right\}.
         \end{equation}
         \item $\widetilde{J}_{\widetilde{\mathcal{A}}}$ is a measurable range function.
     \end{enumerate}
 \end{prop}
\begin{proof}
Let $V= \overline{\Span} \{ \phi(\cdot-n): n \in \mathbb{Z}, \phi \in \mathcal{A} \}$. Then, from Lemma \ref{Vtilde=spantilde}, it follows that $$\widetilde{V}= \overline{\Span} \left\{ \widetilde{\phi(\cdot-n)}: \phi \in \mathcal{A}, n \in \mathbb{Z}\right\}.$$ Define the space $\widetilde{M}:= \widetilde{\Gamma} \widetilde{V}$. Then, using \eqref{201}, it can  be shown that for any $\Phi \in M, \Phi(\xi) \in \widetilde{J}_{\widetilde{A}}(\xi)$, for a.e. $\xi \in [0,n_0] $. Therefore, $M \subset M_{\widetilde{J}_{\widetilde{\mathcal{A}}}}$
. Further, using the assumption that $\mathcal{A} \subset P_k(L^2(\mathbb{R}))$ and the definition of $\widetilde{J}_{\widetilde{\mathcal{A}}}(\xi)$, it can be concluded that 
\begin{equation}\label{JA={0}}
\widetilde{J}_{\widetilde{\mathcal{A}}}(\xi)= \{ \boldsymbol{0} \}, \text{ for a.e. } \xi \in [0,n_0] \setminus [k,k+1].
\end{equation}
Here, $\boldsymbol{0}$ denotes the zero vector in $\mathbb{C} \times_\lambda \ell^2(\mathbb{Z})$.
In order to prove \eqref{VtoVtilde}, take any $\boldsymbol{0} \neq \Psi \in L^2([0,n_0], \mathbb{C} \times_{\lambda} \ell^2(\mathbb{Z}))$ such that $ \Psi \perp M$. Then, for any $\Phi \in \widetilde{\Gamma} \widetilde{\mathcal{A}}$ and $n \in \mathbb{Z}$, we have $e^{- 2 \pi i n \cdot} \Phi(\cdot) \in \widetilde{\Gamma}\widetilde{V}$. Hence,
    \begin{align*}
        0&= \int_0^{n_0} \left \langle e^{- 2 \pi i n \xi}\Phi(\xi), \Psi(\xi) \right \rangle d\xi\\
        &= \int_k^{k+1}e^{-2 \pi i n \xi} \left \langle \Phi(\xi), \Psi(\xi) \right \rangle d\xi.
    \end{align*}
Therefore, all the Fourier coefficients of the function $\xi \mapsto \left \langle \Phi(\xi), \Psi(\xi) \right \rangle$ defined from $[k,k+1]$ to $\mathbb{C}$ vanish. That is,
        \[\left \langle \Phi(\xi), \Psi(\xi) \right \rangle =0, \text{ for a.e. } \xi \in [k,k+1].\]
    Thus, $\Psi(\xi) \in \widetilde{J}_{\widetilde{\mathcal{A}}}(\xi)^{\perp}$ for a.e. $ \xi \in [k,k+1]$. If we further assume that $\Psi \in M_{\widetilde{J}_{\widetilde{\mathcal{A}}}}$, then $\Psi(\xi) \in \widetilde{J}_{\widetilde{\mathcal{A}}}(\xi)$, for a.e. $\xi \in [0,n_0]$. Hence, $\Psi(\xi)= \boldsymbol{0}$ , for a.e. $\xi \in [k,k+1]$, which along with \eqref{JA={0}} implies that $\Psi(\xi)= \left\{ \boldsymbol{0}\right\}$, for a.e. $\xi \in [0,n_0]$. Thus, there does not exist $\boldsymbol{0} \neq \Psi \in M_{\widetilde{J}_{\widetilde{\mathcal{A}}}}$
 which is orthogonal to $M$, and therefore $M= M_{\widetilde{J}_{\widetilde{\mathcal{A}}}}$.\\
 
% Only the proof of the forward implication of Statement (1) is given. 
The proof of the converse of statement $(1),$ 
and of statement $(2)$ are omitted.

    % Let $V=\overline{\Span}\left\{ \phi(\cdot -n): n \in \mathbb{Z}, \phi \in \mathcal{A} \right\}$. Then by \eqref{} we have $\widetilde{V}= \overline{\Span}\left\{ \widetilde{\phi(\cdot-n)}: \phi \in \mathcal{A}, n \in \mathbb{Z} \right\}$. Suppose $M:=\widetilde{\Gamma}\widetilde{V}$, then using \eqref{}, it can be shown that for any $\Phi \in M$, $\Phi(\xi) \in \widetilde{J}_{\widetilde{\mathcal{A}}}(\xi)$ a.e. $\xi \in [0,n_0]$. Therefore $M \subseteq M_{\widetilde{J}_{\widetilde{\mathcal{A}}}}$.
    
%     If we further assume that $\Psi \in M_{\widetilde{J}_{\widetilde{\mathcal{A}}}}$, then $\Psi(\xi) \in \widetilde{J}_{\widetilde{\mathcal{A}}}(\xi)$ a.e. $\xi \in [0,n_0]$, hence we get that $\Psi(\xi)=0$  a.e. $ \xi \in [k,k+1]$. Further, $\widetilde{J}_{\widetilde{\mathcal{A}}}(\xi)=\{0\}
% $ a.e. $\xi \in [0,n_0]
%  [k,k+1]$ by definition. Therefore we can conclude that $\Psi(\xi) = 0$  a.e. $ \xi \in [0,n_0]$, which is a contradiction, which in turn implies that $M=M_{\widetilde{J}_{\widetilde{\mathcal{A}}}}$. The converse can also be shown using a similar technique.
 \end{proof}
 The above proposition is important because it forms the foundation for proving Lemma \ref{308} (as in classical case). The result of Lemma \ref{308} will be used repeatedly throughout the paper. 
\begin{lemma}\label{308}
    Let $(c,f) \in \mathcal{R}_\lambda$ and $ k\in \{0, \ldots, n_0-1\}$. Suppose $\phi_1, \ldots, \phi_l \in P_k(L^2(\mathbb{R}))$ and $V=V(\phi_1, \ldots, \phi_l)$. Then, 
    \begin{equation}\label{1001}
        \left( \widetilde{\Gamma}P_{\widetilde{V}}(c,f) \right)(\xi)= P_{\widetilde{J}_{\widetilde{V}}(\xi)} \left( \widetilde{\Gamma} (c,f) \right)
(\xi) \text{ for a.e. } \xi \in [0,n_0].
\end{equation}
\end{lemma}
% \begin{proof}
%     Using \eqref{} and \eqref{}, we have 
%     \[ \left( P_{M_{\widetilde{J}_{\widetilde{\mathcal{A}}}}}\widetilde{\gamma} \widetilde{f} \right)(\xi)= P_{\widetilde{J}_{\widetilde{A(\xi)\left( \widetilde{\Gamma}\widetilde{f} \right)(\xi)= P_{\widetilde{J}_{\widetilde{V}}(\xi)} \left( \widetilde{\Gamma} \widetilde{f} \right)(\xi) \text{ a.e. } x \in [0,n_0].\]
    
% \end{proof}

 \section{Optimality for the class of FSISs}\label{FSISs}
 
 Let $l \in \mathbb{N}$. Given measurements $\left\{Y^j \right\}_{j=1}^m = \left\{ \left\{ y^1_k \right\}_{k \in \mathbb{Z}}, \ldots, \left\{ y^m_k \right\}_{k \in \mathbb{Z}} \right\} \subset \ell^2(\mathbb{Z})$, here we consider the minimization problem (we make use of the first form, see \eqref{8}) for the class $\mathcal{V}^l$ consists of FSISs of length at most $l$. That is,
 \begin{equation}
     \argmin_{V \in \mathcal{V}^l} \sum_{j=1}^m \left\| Y^j - P_{\widetilde{V}} Y^j \right\|^2_{{\mathcal{R}_\lambda}}, 
  \end{equation}
where $\mathcal{V}^l:= \left\{ V \subset L^2(\mathbb{R}) : V \text{ is an FSIS of length at most } l \right\}$.
\begin{definition}
   Let the map $\pi: \mathbb{Z} \longrightarrow \mathcal{B}(\widetilde{L^2(\mathbb{R})})$ be defined for $l \in \mathbb{Z}$ by 
    \begin{equation}
    \pi(l) \colon \widetilde{L^2(\mathbb{R})} \longrightarrow \widetilde{L^2(\mathbb{R})}, \quad \pi(l)\left(\widetilde{f}\right)= \widetilde{T_lf}.
  \end{equation}
\end{definition}
Then, we have the following lemma.
\begin{lemma}
    The map $\pi$ is a unitary representation of  $\mathbb{Z}$ onto $\mathcal{B}(\widetilde{L^2(\mathbb{R})}).$
\end{lemma}
\begin{proof}
    First, we show that for each $l \in \mathbb{Z}$, $\pi(l)$ is a unitary map. For this, it is enough to prove that $\pi$ is a surjective isometry.

    Let $f \in  \widetilde{L^2(\mathbb{R})}.$ Then,
    \begin{align*}
   \left\| \pi\left(l\right) (\widetilde{f}) \right\|^2&= \left\| \widetilde{T_lf}\right\|^2= \sum_{k \in  \mathbb{Z}} \left| \left(T_lf\right) \left( \frac{k^g}{n_0}  \right) \right|^2 + \lambda\left\| f(\cdot -l ) \right\|^2\\
 &= \sum_{k \in \mathbb{Z}} \left| \left\langle T_lf, g \left( \cdot - \frac{k}{n_0} \right) \right\rangle \right|^2 + \lambda \left\| f(\cdot -l) \right\|^2\\
 &=  \sum_{k \in \mathbb{Z}} \left| \left\langle f, g \left( \cdot - \left( \frac{k-ln_0}{n_0} \right) \right) \right\rangle \right|^2 + \lambda\left\| f(\cdot- l) \right\|^2\\
 &= \sum_{k^{'} \in \mathbb{Z}} \left| \left\langle f, g\left( \cdot- \left( \frac{k^{'}}{n_0} \right) \right) \right\rangle \right|^2 + \lambda\|f\|^2= \left\| \widetilde{f} \right\|^2.
\end{align*}
 Therefore, $\pi(l)$ is an isometry. Further, since $L^2(\mathbb{R})$ is shift-invariant, $\pi(l)$
 is also surjective. Furthermore, it is easy to show that $\pi$ is a homomorphism. 
Hence, $\pi$ is a unitary representation of $\mathbb{Z}$ onto $\mathcal{B}(\widetilde{L^2(\mathbb{R})})$.
\end{proof}

\begin{definition}
    Let $W$ be a closed $\pi(\mathbb{Z})$-invariant subspace of $\mathcal{H}$, i.e. $\pi(k)w \in W$ for all $k \in \mathbb{Z}$ and $w \in W$. The $\pi(\mathbb{Z})$-dimension of $W$ is defined to be the minimal dimension of a subspace V such that
    \[W=\overline{\Span} \left\{ \pi(k)v: k \in \mathbb{Z}, v \in V \right\}.\]
\end{definition}

Let $W \subset \widetilde{L^2(\mathbb{R})}$. Then from the above definition, it is clear that $W$ is a $\pi(\mathbb{Z})$-invariant subspace of dimension less than or equal to $l$ if and only if there exist $\widetilde{\phi_1}, \ldots, \widetilde{\phi_l} \in \widetilde{L^2(\mathbb{R})}$ such that $$W=\overline{\Span} \left\{ \pi(k) \widetilde{\phi_i}: i \in \{1, \ldots, l \}, k \in \mathbb{Z} \right\}.$$ Further, it follows  from Lemma \ref{Vtilde=spantilde} that if $V=V(\phi_1, \ldots, \phi_l) \subset L^2(\mathbb{R})$, then the $\pi(\mathbb{Z})$-invariant subspace generated by $\widetilde{\phi_1}, \ldots, \widetilde{\phi_l}$ is equal to $\widetilde{V}.$ 

% \begin{lemma}\label{801}
% \begin{enumerate}
%     \item Let $W \subset \widetilde{L^2(\mathbb{R})}$. Then  $W$ is $\pi(\mathbb{Z})$-extra invariant subspace of dimension less than or equal to $l$ if and only if there exists $\widetilde{f_1}, \ldots, \widetilde{f_l} \in \widetilde{L^2(\mathbb{R})}$ such that $W=\overline{\Span} \left\{ \pi(k) \widetilde{f_i}: i \in \{1, \ldots, l \}, k \in \mathbb{Z} \}\right\}$.
%     \item Let $f_1, \ldots, f_l \in L^2(\mathbb{R}) \text{ and } V=V(f_1, \ldots, f_l)$. Then, 
% \begin{align*}
%     \overline{\Span} \left\{ \widetilde{f_i(\cdot - k)}: i \in \{1, \ldots, l\}, k \in \mathbb{Z} \right\}&= \left\{ \widetilde{f}: f \in \overline{\Span} \left\{ f_i(\cdot-k): i \in \{1, \ldots, l\}, k \in \mathbb{Z}\right\} \right\}.\\
%  \text{ i.e., the  } \pi(\mathbb{Z}) \text{ invariant space generated by } \widetilde{f_1}, \ldots,\widetilde{f_r} &= \widetilde{V}
%  \end{align*}
%  \end{enumerate}

% \end{lemma}

\begin{theorem}
    Let $l \in \mathbb{N}$. Suppose the measurements $\left\{Y^j\right\}_{j=1}^m=\left\{\left\{y^1_k\right\}_{k \in \mathbb{Z}},\ldots, \left\{y^m_k\right\}_{k \in \mathbb{Z}} \right\} \subset \ell^2(\mathbb{Z})$ are given. Then,
    \begin{equation}
         \argmin_{V \in \mathcal{V}^l} \sum_{j=1}^m \left\| Y^j - P_{\widetilde{V}} Y^j \right\|^2 \text{ (Minimization Problem Form 1)} 
    \end{equation}
    has a minimizer.
\end{theorem}

\begin{proof}
    From our above discussion, we conclude that $V \in \mathcal{V}^l$ if and only if $\widetilde{V}$ belongs to the collection of $\pi(\mathbb{Z})$-invariant subspaces of $\widetilde{L^2(\mathbb{R})}$ of length at most $l$. Now, using \cite[Theorem 3.8]{aldroubifca2011}, we can assume the existence of a minimizer.
 \end{proof}

\section{Optimality for the class of FSISs with extra invariance}\label{FSISwithextrainvariance}
Here, we consider the minimization problem for the class of FSISs with $\frac{\mathbb{Z}}{n_0}$ extra invariance. Recall that $n_0$ is our assumed sampling/measurement rate. 
% i.e. for any $f\in L^2(\mathbb{R})$, our measurements are $S_g^{n_0}(f)=\left\{  \left\langle  f,g\left(\cdot - \frac{k}{n_0}\right)  \right\rangle    \right\}_{k\in \mathbb{Z}}$.

% Let us start by introducing some notation.
Fix $l \in \mathbb{N}$ and let 
    \begin{equation}\label{Vln0}
        \mathcal{V}_{n_0}^l:= \left \{V: V \text{ is an FSIS of length at most } l  \text{ and } V  \text{ is } \frac{\mathbb{Z}}{n_0}\text{ extra invariant}\right \}.
    \end{equation}
     
    Hence our minimization problem (we make use of the second form, see \eqref{minprobform2}) is
    \begin{equation}\label{11}
   \arg\min_{V \in \mathcal{V}_{n_0}^l} \sum_{j=1}^m \left \| \widetilde{f_{Y,j}}- P_{\widetilde{V}} \widetilde{f_{Y,j}} \right \|^2_{\mathcal{R_{\lambda}}}.  
    \end{equation}
    Our first step is to analyse the structure of $\widetilde{V}$ for a given $V \in \mathcal{V}_{n_0}^l$.
Let $V \in \mathcal{V}_{n_0}^l$, then we know  $ V = V_0 \dot{\oplus} \cdots \dot{\oplus} V_{n_0-1}$ (see \eqref{101}). In the following lemma, we prove that $\Tilde{V}$ will have a similar representation. The key relation we use here is \eqref{18}.
    \begin{lemma}\label{decompositionoftilde{V}}
        Let $V \in \mathcal{V}_{n_0}^l$. Then
        \begin{equation}\label{12}
    \widetilde{V}=\widetilde{V_0}\dot{\oplus} \cdots
   \dot{\oplus} \widetilde{V_{n_0-1}},
        \end{equation} where $\{V_i\}_{i=0}^{n_0-1}$ are as defined in \eqref{defnVk}.
    \end{lemma}
    \begin{proof}
    Let $V\in \mathcal{V}_{n_0}^l$. Then from \eqref{101}, we can write
    \[ V = V_0 \dot{\oplus} \cdots \dot{\oplus} V_{n_0-1}.\]
    In fact, for any $f\in V$, we have $f  = f^0 + \cdots + f^{n_0-1}
       $, which implies that $ \widetilde{f} = \widetilde{f^0} + \ldots \widetilde{f^{n_0-1}}.$ As $\widetilde{f^i} \in \widetilde{V_i}$ for all $ i \in \lbrace 0,\ldots , n_0-1 \rbrace$, in order to prove \eqref{12}, it is enough to show that $\left\langle \widetilde{f^i},\widetilde{f^j}\right\rangle_{R_\lambda}=0, \hspace{0.2 cm} \forall \hspace{0.1 cm} i \ne j \in \lbrace 0,\ldots,n_0-1 \rbrace$.\\
       
    Let $i,j \in \lbrace 0,\ldots, n_0-1\rbrace$. Then 
    \begin{align}
        \left\langle \widetilde{f^i},\widetilde{f^j}  \right\rangle & = \left\langle  \widetilde{\Gamma} \widetilde{f^i}, \widetilde{\Gamma} \widetilde{f^j}\right\rangle
        \nonumber \\
        & = \int_0^{n_0} \left\langle  \left(\widetilde{\Gamma} \widetilde{f^i}\right)(\xi), \left(\widetilde{\Gamma} \widetilde{f^j}\right)(\xi)\right\rangle_{\mathbb{C} \times_\lambda \ell^2(\mathbb{Z})} d\xi.
        \label{13}
    \end{align}
    Now, for a.e. $\xi \in [0,n_0]$
    \begin{align*}
        \left( \widetilde{\Gamma}\widetilde{{f^i}} \right) (\xi) & = \left(  \sum_{l\in \mathbb{Z}} \widehat{f^i} \left( \xi + l n_0 \right) \overline{\widehat{g}\left( \xi + l n_0 \right)},\left\lbrace \widehat{f^i}\left( \xi + l n_0 \right) \right\rbrace_{l\in \mathbb{Z}} \right) 
        \\
        & = \left(\sum_{l \in \mathbb{Z}} \widehat{f} \left( \xi + l n_0 \right) \mathcal{X}_{B_i}\left( \xi + l n_0 \right) \overline{\widehat{g}\left( \xi + l n_0 \right)}, \left\lbrace \widehat{f} \left( \xi + l n_0 \right) \mathcal{X}_{B_i}\left( \xi + l n_0 \right) \right\rbrace_{l \in \mathbb{Z}}  \right)
        \\
      \numberthis \label{305}  & = \left( \mathcal{X}_{B_i}\left( \xi \right) \sum_{l \in \mathbb{Z}} \widehat{f} \left( \xi + l n_0 \right) \overline{\widehat{g}\left( \xi + l n_0 \right)}, \left\lbrace  \mathcal{X}_{B_i}\left( \xi \right) \widehat{f} \left( \xi + l n_0 \right)  \right\rbrace_{l \in \mathbb{Z}} \right) 
      = \mathcal{X}_{B_{i}} (\xi)\left(\widetilde{\Gamma}\widetilde{f}\right)(\xi).
    \end{align*}
    Therefore, it follows from \eqref{13} that 
    \begin{align*}
        \left\langle \widetilde{f^i} , \widetilde{f^j} \right\rangle &=  \int_0^{n_0} \Bigg( \mathcal{X}_{B_i}\left( \xi \right)\mathcal{X}_{B_j}\left( \xi \right) \left| \sum_{l \in \mathbb{Z}}\widehat{f} \left( \xi + l n_0 \right) \overline{\widehat{g}\left( \xi + l n_0 \right)} \right|^2 
        \\
        &\quad + \lambda \mathcal{X}_{B_i}\left( \xi \right) \mathcal{X}_{B_j}\left( \xi \right) \sum_{l \in \mathbb{Z}} \left|\widehat{f} \left( \xi + l n_0 \right) \right|^2\bigg) dx
        \\
        & = 0 \hspace{0.1cm} \textnormal{ if } \hspace{0.1 cm} i \ne j,
    \end{align*}
    proving our assertion.
    \end{proof}
    Using the above decomposition of $\widetilde{V}$, we further restate our minimization problem. For each $j\in \lbrace 1,\ldots ,m \rbrace$, we can orthogonally decompose $ \widetilde{f_{Y,j}} $ in the following manner:
    \[  \widetilde{f_{Y,j}} = \widetilde{f^0_{Y,j}} + \cdots + \widetilde{f^{n_0}_{Y,j}}.\]
    Therefore,
    \begin{align*}
        \sum_{j=1}^m \norm{\widetilde{f_{Y,j}}-{P}_{\widetilde{V}} \widetilde{f_{Y,j}}}_{R_\lambda}^2 & = \sum_{j=1}^m \norm{ \sum_{k=0}^{n_0-1} \widetilde{f^k_{Y,j}} -P_{\widetilde{V_0} \dot{\oplus} \cdots \dot{\oplus} \widetilde{V_{n_0-1}}} \sum_{k=0}^{n_0-1} \widetilde{f^k_{Y,j}}}^2
        \\
        & = \sum_{j=1}^m \norm{ \sum_{k=0}^{n_0-1} \widetilde{f^k_{Y,j}} - \sum_{k=0}^{n_0-1} P_{\widetilde{V}_k} \widetilde{f^k_{Y,j}}}^2
        =\sum_{j=1}^m \norm{\sum_{k=0}^{n_0-1}  \left( \widetilde{f^k_{Y,j}} - P_{\widetilde{V}_k} \widetilde{f^k_{Y,j}} \right)    }^2
        \\
        & = \sum_{j=1}^m \sum_{k=0}^{n_0-1}  \norm{ \widetilde{f^k_{Y,j}} - P_{\widetilde{V}_k} \widetilde{f^k_{Y,j}}}^2=\sum_{k=0}^{n_0-1}  \sum_{j=1}^m \norm{ \widetilde{f^k_{Y,j}} - P_{\widetilde{V}_k} \widetilde{f^k_{Y,j}}}^2.
    \end{align*}
    Hence, the minimization problem~\eqref{11} takes the form
    \begin{align}
        \label{14}
        \arg\min_{V \in \mathcal{V}_{n_0}^{l}} \sum_{k=0}^{n_0-1} \sum_{j=1}^m   \norm{ \widetilde{f^k_{Y,j}} - P_{\widetilde{V}_k} \widetilde{f^k_{Y,j}}}^2.
    \end{align}

 In order to solve the above minimization problem, we follow a two-step method. First, we define $n_0$ new minimization problems motivated by the above one. Then, from the solutions of these new problems, we construct a solution for our original minimization problem. 
    \begin{definition} Let $l\in \mathbb{N}.$ For each $k\in \{0,\ldots,n_0-1\}$, define
            \begin{align*}
        \mathcal{V}_{n_0}^{l,k}:=\left\lbrace V \subseteq L^2(\mathbb{R}) : \ \hspace{0.1 cm} V\, \textnormal{is an FSIS of length at most}\, l \, \textnormal{and} \, V \subseteq P_k (L^2(\mathbb{R}))     \right\rbrace.
    \end{align*}
    \end{definition}
\begin{lemma}
    For each $k \in \left\lbrace 0,\ldots, n_0-1 \right\rbrace$, there exists a $\phi_k \in P_k(L^2(\mathbb{R}))$ such that 
    \begin{align}
        \label{19}
        V(\phi_k)=\arg\min_{V \in \mathcal{V}_{n_0}^{l,k}} \sum_{j=1}^m \norm{\widetilde{f^k_{Y,j}} - P_{\widetilde{V}} \widetilde{f^k_{Y,j}}}^2.
    \end{align}
\end{lemma}

\begin{proof} The proof of this lemma follows in a similar way as that of \cite[Theorem 2.1]{aldroubiacha2007}.
Fix $k \in \{0,\ldots, n_{0}-1\}.$
    Let $V \in \mathcal{V}_{n_0}^{l,k}$. Then, using Lemmas \ref{gammaisisom} and \ref{308},
    \begin{align*}
        \sum_{j=1}^m \norm{\widetilde{f^k_{Y,j}} - P_{\widetilde{V}} \widetilde{f^k_{Y,j}}}^2_{R_{\lambda}} & = \sum_{j=1}^m \norm{\widetilde{\Gamma} \widetilde{f^k_{Y,j}} - \widetilde{\Gamma} P_{\widetilde{V}} \widetilde{f^k_{Y,j}}}^2_{L^2\left([0,n_0],\mathbb{C} \times_{\lambda} \ell^2(\mathbb{Z}) \right)}
        \\
        & = \sum_{j=1}^m \int_0^{n_0} \norm{\left( \widetilde{\Gamma} \widetilde{f^k_{Y,j}} \right)(\xi) - \left( \widetilde{\Gamma} P_{\widetilde{V}}\widetilde{f^k_{Y,j}} \right)(\xi)}^2_{\mathbb{C} \times_{\lambda} \ell^2(\mathbb{Z})} d\xi
        \\
        & = \sum_{j=1}^m \int_0^{n_0} \norm{ \left( \widetilde{\Gamma} \widetilde{f^k_{Y,j}}\right)(\xi)- P_{\widetilde{J}_{\widetilde{V}}(\xi)} \left( \widetilde{\Gamma} \widetilde{f^k_{Y,j}}\right)(\xi)}^2 d\xi
        \\
     \numberthis \label{21}   & = \int_0^{n_0} \sum_{j=1}^m  \norm{\left( \widetilde{\Gamma} \widetilde{f^k_{Y,j}}\right)(\xi) - P_{\widetilde{J}_{\widetilde{V}}(\xi)}\left( \widetilde{\Gamma} \widetilde{f^k_{Y,j}}\right)(\xi) }^2 d\xi.
    \end{align*}
For a.e. $\xi \in [0, n_0]$, define  $\mathcal{F}_{k,\xi} := \left\{ \left( \widetilde{\Gamma}\widetilde{f_{Y,1}^{k}} \right)(\xi), \ldots, \left( \widetilde{\Gamma}\widetilde{f_{Y,m}^{k}} \right)(\xi) \right\}$. Now using \eqref{301}, \eqref{44} and \eqref{305},
\begin{align*}
    \mathcal{B}\left( \mathcal{F}_{k, \xi} \right)_{ij}&=\left\langle  \left( \widetilde{\Gamma} \widetilde{f^k_{Y,i}} \right) (\xi), \left( \widetilde{\Gamma} \widetilde{f^k_{Y,j}} \right) (\xi) \right\rangle\\
    &= \left\langle \mathcal{X}_{B_k}(\xi) \left(\widetilde{\Gamma} \widetilde{f_{Y,i}}\right)(\xi), \mathcal{X}_{B_k}(\xi) \left(\widetilde{\Gamma} \widetilde{f_{Y,j}}\right)(\xi) \right\rangle\\
    &=\mathcal{X}_{B_k}(\xi) \left\langle \sum_{n \in \mathbb{Z}} y^i_n e^{- \frac{ 2 \pi i n \xi}{n_0}}\left( \widetilde{\Gamma} \widetilde{f} \right)(\xi), \sum_{m \in \mathbb{Z}}y^j_m e^{- \frac{
 2 \pi i m \xi}{
 n_0}} \left( \widetilde{\Gamma} \widetilde{f} \right)(\xi) \right\rangle\\
 &=\mathcal{X}_{B_k} (\xi) \sum_{n \in \mathbb{Z}} y^i_n e^{- \frac{2 \pi i n \xi}{n_0}} \overline{\sum_{m \in \mathbb{Z}} y^j_m e^{- \frac{2 \pi i m \xi}{n_0}}} \left\| \left( \widetilde{\Gamma} \widetilde{f} \right)(\xi) \right\|^2.
\end{align*}
   
   Assume that the eigenvalues of the matrix $\mathcal{B}\left( \mathcal{F}_{k, \xi} \right)$  are  $\lambda_1^k(\xi)\ge \cdots\ge \lambda_m^k(\xi)\ge 0.$
Let $U_k(\xi)$ be the measurable $m\times m$ matrix as in \eqref{309}. Since $\mathcal{B}\left( \mathcal{F}_{k, \xi} \right)$ is  $n_0 \mathbb{Z}$-periodic on $\mathbb{R},$ we choose
$U_{k}(\xi)$ also to be $n_0 \mathbb{Z}$-periodic. Let $U_i^k(\xi)$ denote the $i$th row of $U_k(\xi)$. Then $z^k_i(\xi)=\left(z^k_{i,1}(\xi),\cdots,z^k_{i,m}(\xi)\right):=U_i^k(\xi)^*$ is the left eigenvector of 
$\mathcal{B}\left( \mathcal{F}_{k, \xi} \right)$ with eigenvalue $\lambda_i^k(\xi)$ for all $i\in \{1,\ldots,m\}$. For each $i\in \{1,\ldots,l\}$, define $q_i^k(\xi) \in \mathbb{C}\times \ell^2(\mathbb{Z})$ as 
\begin{equation}
    q^k_i(\xi)=\widetilde{\sigma_i^k}(\xi)\sum_{j=1}^mz^k_{i,j}(\xi)\left(\widetilde{\Gamma}
    \widetilde{f_{Y,j}^{k}}\right)(\xi),
\end{equation}
 where $\widetilde{\sigma_i^k}(\xi)=(\lambda_i^k)^{-\frac{1}{2}}(\xi)\quad \text{if}\quad \lambda_i^k(\xi)\neq 0$
and $ \widetilde{\sigma_i^k}(\xi)=0$ otherwise. From Theorem \ref{303}, it follows that the space $S^k_\xi=\Span\left\{q_1^k(\xi),\ldots,q_l^k(\xi)\right\}$ satisfies
\begin{align*}
    \sum_{j=1}^m \left\| \left( \widetilde{\Gamma} \widetilde{f^k_{Y,j}} \right)(\xi) -P_{S^k_\xi} \left( \widetilde{\Gamma} \widetilde{f^k_{Y,j}} \right)(\xi) \right\|^2 &\leq \sum_{j=1}^m \left\| \left( \widetilde{\Gamma} \widetilde{f^k_{Y,j}} \right)(\xi) -P_{\widetilde{J}_{\widetilde{V}}(\xi)} \left( \widetilde{\Gamma} \widetilde{f^k_{Y,j}} \right) (\xi) \right\|^2, \text{ for a.e. } \xi \in [0,n_0].
    \end{align*}
    That is,
    \begin{align*}
      \numberthis \label{**} \quad \sum_{j=1}^m \left\|P_{S^k_\xi} \left( \widetilde{\Gamma} \widetilde{f^k_{Y,j}} \right) (\xi) \right\|^2 &\geq \sum_{j=1}^m \left\| P_{\widetilde{J}_{\widetilde{V}}}(\xi) \left( \widetilde{\Gamma} \widetilde{f^k_{Y,j}} \right) (\xi) \right\|^2, \text{ for a.e. } \xi \in [0,n_0].
\end{align*} Moreover, using \eqref{44} and \eqref{305}, we get that  for a.e. $\xi\in [0,n_0]$ and all $i \in \{1, \ldots,l\}$, 
\begin{align*}
    q_i^k(\xi)&= \widetilde{\sigma_i^k}(\xi)\sum_{j=1}^mz^k_{i,j}(\xi)\chi_{[k,k+1]}(\xi)\sum_{n\in \mathbb{Z}}y_n^je^{-\frac{2\pi i n\xi}{n_0}}\left(\widetilde{\Gamma}\widetilde{f}\right)(\xi)\\
    &= \alpha^k_i(\xi) \left( \widetilde{\Gamma} \widetilde{f} \right)(\xi),
\end{align*}
  where
    \[\alpha^k_i(\xi):=\widetilde{\sigma_i^k}(\xi)\sum_{j=1}^m z^k_{i,j}(\xi) \mathcal{X}_{[k,k+1]}(\xi) \sum_{n \in \mathbb{Z}} y^j_n e^{ - \frac{
    2 \pi i n \xi}{n_0}}.\]
Hence,
\begin{align*}
    S_\xi^k&=\Span\left\{ \alpha_1^k(\xi)(\widetilde{\Gamma}\widetilde{f})(\xi),\ldots,\alpha_l^k(\xi)(\widetilde{\Gamma}\widetilde{f})(\xi)\right\}= \Span\left\{ \chi_{\widetilde{C}_k}(\xi)(\widetilde{\Gamma}\widetilde{f})(\xi)\right\},
\end{align*}
where 
\begin{equation}
    \widetilde{C}_k:=\{\xi\in [0,n_0] : \exists \;i\in\{1,\cdots,l\}\;\text{such that}\;\alpha_i^k(\xi)\ne 0\}.
\end{equation}
Clearly, $\widetilde{C}_k$ forms a measurable set.
Defining $C_k= \cup_{j \in \mathbb{Z}} \left( \widetilde{C_k}+n_0j \right)$ and $\widehat{\phi_k}= \mathcal{X}_{C_k} \widehat{f}$, we can conclude that for a.e. $\xi \in [0,n_0], S^k_\xi= \Span \left\{ \left( \widetilde{\Gamma} \widetilde{\phi_k} \right)(\xi) \right\}= \widetilde{J}_{\widetilde{\phi_k}}(\xi).$
 Indeed, for a.e. $\xi \in [0,n_0]$,
 \begin{align*}
\left( \widetilde{\Gamma} \widetilde{\phi_k} \right)(\xi)&= \left( \sum_{l \in \mathbb{Z}} \widehat{\phi_k}(\xi+ln_0) \overline{\widehat{g}(\xi+ln_0)}, \left\{ \widehat{\phi_k}(\xi+ln_0) \right\}_{l \in \mathbb{Z}} \right)\\
&= \left( \sum_{l \in \mathbb{Z}} \mathcal{X}_{C_k}(\xi) \widehat{f}(\xi+ln_0) \overline{\widehat{g}(\xi+ln_0)}, \left\{ \mathcal{X}_{C_k}(\xi) \widehat{f}(\xi+ln_0) \right\}_{l \in \mathbb{Z}} \right)\\
&= \mathcal{X}_{\widetilde{C_k}}(\xi) \left( \widetilde{\Gamma}\widetilde{f}\right)(\xi).
\end{align*}
Besides, from Lemma \ref{JtildeVtilde=JtildeAtilde}, we know that $\widetilde{J}_{\widetilde{V(\phi_k)}}(\xi)= \Span \left\{ \left( \widetilde{\Gamma} \widetilde{\phi_k}  \right)(\xi) \right\}$, for a.e. $\xi \in [0,n_0]$. Hence, using the above observations and \eqref{**}, we get
\begin{align*}
\numberthis \label{mainJphik(xi)inequality}  \sum_{j=1}^m\left\|P_{\widetilde{J}_{\widetilde{V(\phi_k)}}}(\xi) \left( \widetilde{\Gamma} \widetilde{f^k_{y,j}} \right)(\xi) \right\|^2& \geq \sum_{j=1}^m \left\| P_{\widetilde{J}_{\widetilde{V}}(\xi)} \left( \widetilde{\Gamma} \widetilde{f^k_{Y,j}} \right) (\xi) \right\|^2\\
\implies \sum_{j=1}^m \left\| \left( \widetilde{\Gamma} \widetilde{f^k_{Y,j}} \right)(\xi) - P_{\widetilde{J}_{\widetilde{V(\phi_k)}}(\xi)} \left( \widetilde{\Gamma} \widetilde{f^k_{Y,j}} \right) (\xi) \right\|^2& \leq \sum_{j=1}^m \left\| \left( \widetilde{\Gamma} \widetilde{f^k_{Y,j}} \right) (\xi) - P_{\widetilde{J}_{\widetilde{V}}(\xi)} \left( \widetilde{\Gamma} \widetilde{f^k_{Y,j}} \right) (\xi) \right\|^2.
\end{align*}
Retracing the steps leading to \eqref{21}, we can conclude that $V(\phi_k)$ is a minimizer to our problem.

% Defining, 
% \begin{enumerate}
%     \item $C_k=\cup _{j\in \mathbb{Z}}(\widetilde{C}_k+n_0j)$ and
%     \item $\widehat{\phi}_k(\xi)=\mathcal{X}_{C_k}(\xi)\widehat{f}(\xi)$ a.e. $\xi \in \mathbb{R}$,
% \end{enumerate}
% we can conclude that for a.e. $\xi\in [0,n_0]$,
% \begin{align*}
%     S_\xi^k&=\Span\{(\widetilde{\Gamma}\widetilde{\phi}_k)(\xi)\}=\widetilde{J}_{\widetilde{\phi}_k}(\xi).
% \end{align*}
% Moreover, by Lemma \ref{306}, we can conclude that 
% $$\widetilde{J}_{\widetilde{V(\phi_k)}}(\xi)=\Span\{(\widetilde{\Gamma}\widetilde{\phi_k})(\xi)\}\quad \text{a.e.}\;\xi\in [0,n_0].$$
% Hence using Proposition \ref{307}, it follows that $V(\phi_k)$ is the minimizing shift-invariant space. That is, 
% \begin{equation*}\label{510}
%   \sum_{j=1}^m \left\|P_{\widetilde{J}_{\widetilde{V(\phi_k)}}(\xi)} \left( \widetilde{\Gamma} \widetilde{f^k_{Y,j}} \right) (\xi) \right\|^2 \geq \sum_{j=1}^m \left\| P_{\widetilde{J}_{\widetilde{V}}}(\xi) \left( \widetilde{\Gamma} \widetilde{f^k_{Y,j}} \right) (\xi) \right\|^2 \text{ a.e. } \xi \in [0,n_0].  
% \end{equation*}
%  The interesting point here is that although we minimized over shift-invariant spaces of length at most $l,$ the minimizing space has length at most one.
\end{proof}
\begin{remark}
The interesting fact in the above result is that although we minimized over shift-invariant spaces of length at most $l,$ the minimizing space has length at most one.
\end{remark}

\subsection{Analysis of the length of an FSIS}\label{analysisoflength}

In this subsection, we introduce a new formula for calculating the length of an FSIS $V$, utilizing our fiber map $\widetilde{\Gamma}$. This formula will be essential in proving the main results in both this section and the next.
   
   Let $\phi_1,\ldots,\phi_l\in L^2(\mathbb{R}),\quad V =V(\phi_1,\ldots,\phi_l)$ and $\{V_k\}_{k=0}^{n_0-1}$ be as defined in \eqref{defnVk} for the FSIS $V.$ For a.e. $\xi\in [0,n_0],$
%    Then, by ??, for a.e. $\xi\in [0,1],$
% \begin{enumerate}
%     \item $J_V(\xi)=\Span\{\{\widetilde\phi_i(\xi+n)\}_{n\in \mathbb{Z}}: i\in \{1,\cdots,l\}\}$

%     \item  $J_{V_k}(\xi)=\Span\{\{\widetilde\phi^k_i(\xi+n)\}_{n\in \mathbb{Z}}: i\in \{1,\cdots,l\}\}$ for all $k\in \{0,\cdots,n_0-1\}.$
%     \item using cite
%     \begin{equation}
%         J_V(\xi)=J_{V_0}(\xi)\oplus\cdots\oplus J_{V_{n_0-1}}(\xi).
%     \end{equation}
% \end{enumerate}

% \begin{align*}
%     \widetilde{J}_{\widetilde{V_k}}(\xi)&=\Span \left\{ \left(\widetilde{\Gamma}\widetilde{\phi_1^k} \right)(\xi),\ldots, \left(\widetilde{\Gamma}\widetilde{\phi_l^k} \right)(\xi)\right\}
%     \\
%     &=
%     \Span \left\{ \left(\sum_{m\in \mathbb{Z}}\widetilde{\phi_1^k})(\xi+mn_0)\overline{\widehat{g}(\xi+mn_0)},\{\widetilde{\phi_1^k})(\xi+mn_0)\}_{m\in \mathbb{Z}}\right)\cdots\\&,\left(\sum_{m\in \mathbb{Z}}\widetilde{\phi_l^k})(\xi+mn_0)\overline{\widehat{g}(\xi+mn_0)},\{\widetilde{\phi_l^k})(\xi+mn_0)\}_{m\in \mathbb{Z}}\right) \right\}
% \end{align*}

\begin{align*}
    \widetilde{J}_{\widetilde{V_k}}(\xi)&=\Span \left\{ \left(\widetilde{\Gamma}\widetilde{\phi_1^k} \right)(\xi),\ldots, \left(\widetilde{\Gamma}\widetilde{\phi_l^k} \right)(\xi)\right\}\\
&=\Span \bigg\{  &&\hspace{-7.6cm}\left(\sum_{m\in \mathbb{Z}}\widehat{\phi_1^k}(\xi+mn_0)\overline{\widehat{g}(\xi+mn_0)}, \left\{\widehat{\phi_1^k}(\xi+mn_0)\right\}_{m\in \mathbb{Z}}\right), \ldots,\\
& &&\hspace{-7.6cm}\left(\sum_{m\in \mathbb{Z}}\widehat{\phi_l^k}(\xi+mn_0)\overline{\widehat{g}(\xi+mn_0)},\left\{\widehat{\phi_l^k}(\xi+mn_0)\right\}_{m\in \mathbb{Z}}\right)\bigg\}.
\end{align*}
Note that, by definition, $\widetilde{J}_{\widetilde{V_k}}(\xi)=\{\boldsymbol{0}\}$, for a.e. $\xi\in [0,n_0]\setminus [k,k+1].$ 
Now fix $k=0$, and consider $J_{V_k}(\xi)$, i.e.,  $J_{V_0}(\xi)$. 
For a.e. $\xi\in [0,1]$, $J_{V_0}(\xi)$ has the following structure.
\begin{align*}
    J_{V_0}(\xi)= \Span \bigg\{ &\left( \ldots, \widehat{\phi^0_1}(\xi-n_0) , \ldots, \widehat{\phi^0_1}(\xi-1), \widehat{\phi^0_1}(\xi), \widehat{\phi^0_1}(\xi+1), \ldots, \widehat{\phi^0_1}(\xi+n_0), \ldots \right),\\
    & \hspace{4cm}\vdots\\
    &\left( \ldots, \widehat{\phi^0_l}(\xi-n_0) , \ldots, \widehat{\phi^0_l}(\xi-1), \widehat{\phi^0_l}(\xi), \widehat{\phi^0_l}(\xi+1), \ldots, \widehat{\phi^0_l}(\xi+n_0), \ldots \right)\bigg\}\\
   \numberthis \label{601} =\Span \bigg\{ &\left( \ldots, \widehat{\phi^0_1}(\xi-n_0),0 , \ldots, 0, \widehat{\phi^0_1}(\xi),0, \ldots, \widehat{\phi^0_1}(\xi+n_0), \ldots \right),\\
    & \hspace{4cm}\vdots\\
    &\left( \ldots, \widehat{\phi^0_l}(\xi-n_0) ,0, \ldots, 0, \widehat{\phi^0_l}(\xi), 0, \ldots,0, \widehat{\phi^0_l}(\xi+n_0), \ldots \right)\bigg\}.
\end{align*}

Since $\sum_{k \in \mathbb{Z}}\widehat{\phi^0_i} (\xi+k n_0) \overline{\widehat{g}(\xi+k n_0)}$ is a linear combination of $\left\{ \widehat{\phi^0_i}(\xi+ kn_0) \right\}_{k \in \mathbb{Z}}$, we get 
\begin{align*}
   & \hspace{-4cm}\Dim  \left( \Span \left\{ \left( \sum_{k \in \mathbb{Z}} \widehat{\phi^0_i} (\xi+ k n_0) \overline{\widehat{g}(\xi+k n_0)}, \left\{ \widehat{\phi^0_i}(\xi+ k n_0) \right\}_{k \in \mathbb{Z}} \right): i \in \{1, \ldots, l\} \right\} \right)\\
    &= \Dim \left( \Span \left\{ \left\{ \widehat{\phi^0_i}(\xi+k n_0) \right\}_{k \in \mathbb{Z}}: i \in \{1, \ldots, l \} \right\} \right)\\
    \numberthis \label{602} &= \Dim \bigg( \Span \bigg\{ \left( \ldots, \widehat{\phi^0_1}(\xi- n_0), \widehat{\phi^0_1}(\xi), \widehat{\phi^0_1}(\xi+ n_0), \ldots \right)\\
    & \hspace{4cm}\vdots\\
    & \hspace{2.7cm}\left( \ldots, \widehat{\phi^0_l}(\xi- n_0), \widehat{\phi^0_l}(\xi), \widehat{\phi^0_l}(\xi+ n_0), \ldots \right) \bigg\} \bigg).
\end{align*}
Therefore, from \eqref{601} and \eqref{602}, we can conclude that 
\[\Dim(J_{V_0}(\xi))=\Dim (\widetilde{J}_{\widetilde{V_0}}(\xi)).\]
Similarly, it can be shown  for a.e. $\xi\in [0,1],$
\begin{align*}
\Dim(J_{V_1}(\xi)) &= \Dim (\widetilde{J}_{\widetilde{V_1}}(\xi+1)), \\
\vdots \\
\Dim(J_{V_{n_0-1}}(\xi)) &= \Dim (\widetilde{J}_{\widetilde{V_{n_0-1}}}(\xi+n_0-1)).
\end{align*}

Hence, using \eqref{Boorformulaforlengthequation} we get 
\begin{align*}
    \Len(V)
    &=\esssup_{\xi\in [0,1]} \Dim(J_V(\xi))\\ &=\esssup_{\xi\in [0,1]} \Dim(J_{V_0}(\xi)\dot{\oplus}\cdots\dot{\oplus} J_{V_{n_0-1}}(\xi))\\ &=\esssup_{\xi\in [0,1]} \left(\Dim(J_{V_0}(\xi))+\cdots+ \Dim(J_{V_{n_0-1}}(\xi))\right) \\ 
    % &=\esssup_{\xi\in [0,1]} \Dim(\widetilde{J}_{\widetilde{V_0}}(\xi))+\cdots+ \Dim(J_{\widetilde{V_{n_0-1}}}(\xi))
    % \\ 
     % \numberthis \label{25}
    % &=\esssup_{\xi\in [0,1]} \Dim(\widetilde{J}_{\widetilde{V_0}}(\xi))+\cdots+ \Dim(\widetilde{J}_\widetilde{{V_{n_0-1}}(\xi+n_0-1))}.
 \numberthis \label{newdimensionformula}&= \esssup_{\xi\in [0,1]} \left(\Dim\left(\widetilde{J}_{\widetilde{V_0}}(\xi)\right) + \cdots + \Dim\left(\widetilde{J}_{\widetilde{V_{n_0-1}}}(\xi+n_0-1)\right)\right).
\end{align*}

As $V=V(\phi_1, \ldots, \phi_l)$, the length of $V$ is at most $l$. From  \eqref{newdimensionformula}, we get, for a.e. $\xi \in [0,1],$
\[\Dim\left(\widetilde{J}_{\widetilde{V_0}}(\xi)\right)+\cdots+ \Dim\left(J_{\widetilde{V_{n_0-1}}}(\xi+n_0-1)\right)\le l.\]
That is, for a.e. $\xi \in [0,1]$, the set of indices $A^V_\xi$ defined by 
\begin{equation}\label{AVxi}
    A^V_\xi:= \left\{ k \in \{ 0, \ldots, n_0-1 \}: \widetilde{J}_{\widetilde{V_k}}(\xi+k) \neq \{ \boldsymbol{0} \} \right\}
\end{equation}
has cardinality less than or equal to $l$.
\subsection{The Main result}
\begin{theorem}
Let $l  \in \mathbb{N}$. Suppose the measurements $\left\{Y^j \right\}_{j=1}^m= \left\{ \left\{y^1_k \right\}_{\in \mathbb{Z}}, \ldots, \left\{ y^m_k \right\}_{k \in \mathbb{Z}} \right\} \subset \ell^2(\mathbb{Z})$ are given. Further, let $\mathcal{V}^l_{n_0}$  be as defined in \eqref{Vln0}. Then there exists an FSIS $W \in \mathcal{V}^l_{n_0}$ such that 
\begin{equation} \label{minproblem}
W= \argmin_{V \in \mathcal{V}^l_{n_0}}\sum_{j=1}^m \left\| \widetilde{f_{Y,j}} - P_{\widetilde{V}} \widetilde{f_{Y,j}} \right\|^2 \text{(Minimization Problem Form 2)} .
\end{equation}
\end{theorem}
\begin{proof}
Let $V \in \mathcal{V}^l_{n_0}$. Then, $V=V_0\dot{\oplus} \cdots \dot{\oplus} V_{n_0-1}$
 (see \eqref{101}). Further, by definition of $V_i$'s, $V_i \in \mathcal{V}^{l,i}_{n_0}$, for all $ i \in \{1, \ldots, n_0-1 \}$. Now, as we have already shown, the minimization problem 
\eqref{minproblem} can be restated as \eqref{14}. Therefore, using \eqref{1001}, 
 \begin{align*}
\sum_{k=0}^{n_0-1}\sum_{j=1}^m  \left\| \widetilde{f^k_{Y,j}} - P_{\widetilde{V_k}} \widetilde{f^k_{Y,j}} \right\|^2&= \sum_{k=0}^{n_0-1} \sum_{j=1}^m \int_0^{n_0} \left\| \left( \widetilde{\Gamma} \widetilde{f^k_{Y,j}} \right)(\xi)- \left( \widetilde{\Gamma} P_{\widetilde{V_k}} \widetilde{f^k_{Y,j}} \right) (\xi) \right\|^2 d \xi\\
&= \sum_{k=0}^{n_0-1} \sum_{j=1}^m \int_0^{n_0} \left\| \left( \widetilde{\Gamma} \widetilde{f^k_{Y,j}} \right)(\xi)- P_{\widetilde{J}_{\widetilde{V_k}}(\xi)} \left( \widetilde{\Gamma} \widetilde{f^k_{Y,j}} \right) (\xi) \right\|^2 d\xi\\
\numberthis \label{ktok+1} &= \sum_{k=0}^{n_0-1} \sum_{j=1}^m \int_k^{k+1} \left\| \left( \widetilde{\Gamma} \widetilde{f^k_{Y,j}} \right)(\xi) - P_{\widetilde{J}_{\widetilde{V_k}}(\xi)} \left( \widetilde{\Gamma} \widetilde{f^k_{Y,j}} \right)(\xi) \right\|^2 d\xi\\
\numberthis \label{finalform} &= \sum_{k=0}^{n_0-1} \sum_{j=1}^{m} \int_0^1 \left\| \left(\widetilde{\Gamma} \widetilde{f^k_{Y,j}} \right)(\xi+k) - P_{\widetilde{J}_{\widetilde{V_k}}(\xi+k)} \left( \widetilde{\Gamma} \widetilde{f^k_{Y,j}} \right) (\xi+k) \right\|^2 d\xi.
\end{align*}
The equality  \eqref{ktok+1} follows from the fact that  $\left( \widetilde{\Gamma} \widetilde{f^k_{Y,j}} \right)(\xi)= \{ \boldsymbol{0} \}$, for a.e.  $ \xi \in [0,n_0] \setminus [k,k+1]$.
Moreover, it can be observed that the above statements hold for any FSISs, which is $\frac{\mathbb{Z}}{n_0}$ invariant, without any assumptions on its length. Using \eqref{finalform}, we can restate our minimization problem \eqref{minproblem}, as the following maximization problem:
  \begin{equation}\label{28}
       \argmax_{ V \in \mathcal{V}^l_{n_0}} \sum_{k=0}^{n_0-1} \int_0^1 \sum_{j=1}^m \left\| P_{\widetilde{J}_{\widetilde{V}_k}(\xi+k)} \left( \widetilde{\Gamma} \widetilde{f^k_{Y,j}} \right) (\xi+k) \right\|^2 d\xi.
  \end{equation}
  
For each $k \in  \{0, \ldots, n_0-1 \}$, let $V(\phi_k)$ be as defined in \eqref{19}. Further, let $U:= V(\phi_1) \dot{\oplus} \cdots \dot{\oplus} V(\phi_{n_0-1}),$
 and $U_k:=P_k(U)$, $ \forall\, k \in \{0, \ldots, n_0-1 \}$. Then clearly $U_k=V(\phi_k)$ and $U_k$ is $\frac{\mathbb{Z}}{n_0}$ extra invariant for each $k \in \{0, \ldots, n_0-1\}$. \\

As mentioned above, \eqref{finalform} is true for any $\frac{\mathbb{Z}}{n_0}$-extra invariant FSIS V. Therefore, choosing $V=U$ in \eqref{finalform}, we get 
\begin{equation}\label{formulaforU}
\sum_{j=1}^m \left\| \widetilde{f_{Y,j}} - P_{\widetilde{U}} \widetilde{f_{Y,j}} \right\|^2= \sum_{k=0}^{n_0-1} \int_0^1 \sum_{j=1}^m \left\| \left( \widetilde{\Gamma} \widetilde{f^k_{Y}} \right)(\xi+k) - P_{\widetilde{J}_{\widetilde{U_k}}(\xi+k)} \left( \widetilde{\Gamma} \widetilde{f^k_{Y,j}} \right)(\xi+k) \right\|^2 d \xi.
\end{equation}
Now, using \eqref{mainJphik(xi)inequality}, for a.e. $\xi \in [0,1]$,
\begin{align*}
    \sum_{j=1}^m \left\| P_{\widetilde{J}_{\widetilde{V_k}}(\xi+k)} \left( \widetilde{\Gamma} \widetilde{f^k_{Y,j}} \right) (\xi+k) \right\|^2 & \leq \sum_{j=1}^m \left\| P_{\widetilde{J}_{\widetilde{U_k}}(\xi+k)} \left( \widetilde{\Gamma} \widetilde{f^k_{Y,j}} \right) (\xi+k) \right\|^2 \hspace{0.2cm} \forall \hspace{0.1cm} k \in \{0,\ldots,n_0-1\}.\\
        \numberthis \label{comparisonforUwithoutintegral} \implies \sum_{k=0}^{n_0-1} \sum_{j=1}^m \left\| P_{\widetilde{J}_{\widetilde{V_k}}(\xi+k)} \left( \widetilde{\Gamma} \widetilde{f^k_{Y,j}} \right)(\xi+k) \right\|^2  &\leq \sum_{k=0}^{n_0-1} \sum_{j=1}^m \left\|P_{\widetilde{J}_{\widetilde{U_k}}(\xi+k)} \left( \widetilde{\Gamma} \widetilde{f^k_{Y,j}} \right) (\xi+k) \right\|^2 \\
    \numberthis \label{comparisonforU} \implies \sum_{k=0}^{n_0-1} \int_0^1 \sum_{j=1}^m \left\| P_{\widetilde{J}_{\widetilde{V_k}}(\xi+k)} \left( \widetilde{\Gamma} \widetilde{f^k_{Y,j}} \right)(\xi+k) \right\|^2 d\xi &\leq \sum_{k=0}^{n_0-1} \int_0^1 \sum_{j=1}^m \left\|P_{\widetilde{J}_{\widetilde{U_k}}(\xi+k)} \left( \widetilde{\Gamma} \widetilde{f^k_{Y,j}} \right) (\xi+k) \right\|^2 d\xi.
\end{align*}
Therefore, using \eqref{formulaforU} and \eqref{comparisonforU}, we can conclude that for the case $n_0 \leq l$, $U$ turns out to be a maximizer of \eqref{minproblem}.

If $n_0>l$, then we will construct a new FSIS $W$ from $U$ of length less than or equal to $l$. For this, we will exploit the fact that the cardinality of $A^V_\xi$ (as defined in \eqref{AVxi}) is less than or equal to $l$, for a.e. $\xi \in [0,1]$. That is, $\widetilde{J}_{\widetilde{V_k}}(\xi+k) \neq \{\boldsymbol{0}\}$ for at most $l$ distinct $k \in \{0, \ldots, n_0-1 \}$. This, in turn, implies that the first summation in the left-hand side of the inequality \eqref{comparisonforUwithoutintegral}
 is actually over at most $l$ non-zero terms. Keeping the motivation based on these observations, we construct $W_k$ from $U_k$
 such that $\widetilde{J}_{\widetilde{W_k}}(\xi+k) $ is non-trivial
  for at most $l$ distinct $k \in\{0, \ldots, n_0-1\}$. This is carried out under the constraint that \eqref{comparisonforUwithoutintegral} is maintained (up to sets of measure zero) when $U_k$ is replaced by $W_k$.
  \smallskip

For a.e. $\xi \in [0,1]$, consider the ordered collection 
\[ \left\{ \sum_{j=1}^m \left\| P_{\widetilde{J}_{\widetilde{U_k}}(\xi+k)} \left( \widetilde{\Gamma} \widetilde{f^k_{Y,j}} \right) (\xi+k) \right\|^2 \right\}_{k=0}^{n_0-1}. \]

Select the \( l \) largest terms from the set above. In cases of ambiguity due to equal terms, choose those with the smallest indices. Define the ordered set \( D_\xi \subset \{0, \dots, n_0-1 \} \) as the collection of their indices.
Further, for each $i \in \{0, \ldots, n_0-1 \}$, we define the following. 
\begin{enumerate}
\item $H_i= \{ i+ \xi : \xi \in [0,1] \text{ and } i \in D_\xi \}$. It can be shown that $H_i$ is measurable.
\item $E_i= \cup_{j \in \mathbb{Z}}  (H_i +n_0j).$
\item $\widehat{\psi_i}(\xi)= \mathcal{X}_{E_i}(\xi) \widehat{\phi_i}(\xi)$ for a.e. $\xi \in\mathbb{R}$. Clearly, $ \psi_i \in L^2(\mathbb{R})$.
\end{enumerate}

Let $W_i:= V(\psi_i)$, $ \forall\, i \in \{0, \ldots, n_0-1\}$ and $W=W_0 \dot{\oplus}
 \cdots \dot{\oplus} W_{n_0-1}$. Then, \textbf{we claim that $W$ is the maximizer to our maximization problem \eqref{28}}, We prove our claim in two parts.
\smallskip First, we show that $W \in \mathcal{V}^l_{n_0}$. Since $\supp\widehat{\psi}_i \subset B_i$ for all $ i \in \{0, \dots, n_0-1\}$, the FSISs $\{W_i\}_{i=0}^{n_0-1}$  are $\frac{\mathbb{Z}}{n_0}$-extra invariant, and hence, one can show that $W$ is $\frac{\mathbb{Z}}{n_0}$-extra invariant . Now, we calculate the length of $W$. From \eqref{newdimensionformula} 
    \begin{align*}
        \Len W=& \esssup_{\xi \in [0,1]}\left( \Dim \left( \widetilde{J}_{\widetilde{W_0}} (\xi) \right) + \cdots + \Dim \left( \widetilde{J}_{\widetilde{W_{n_0-1}}} (\xi+ n_0-1) \right)\right)\\
        =& \esssup_{\xi \in [0,1]}\Bigg( \Dim \left( \Span \left\{ \left( \sum_{l \in \mathbb{Z}} \widehat{\psi_0} (\xi+ l n_0) \overline{\widehat{g}(\xi+l n_0)}, \left\{ \widehat{\psi_0}(\xi+l n_0) \right\}_{l \in \mathbb{Z}} \right) \right\} \right) + \cdots  \\
    & +\Dim \left( \Span \left\{ \left( \sum_{l \in \mathbb{Z}} \widehat{\psi_{n_0-1}} (\xi+ n_0-1+ l n_0)\overline{\widehat{g}(\xi+n_0-1+ln_0)}, \left\{ \widehat{\psi_{n_0-1}}(\xi+ n_0 -1+ ln_0) \right\}_{l \in \mathbb{Z}} \right) \right\} \right)\Bigg)\\
         \numberthis \label{30} =&\esssup_{\xi \in [0,1]} \Bigg(\Dim \left( \Span \left\{ \mathcal{X}_{H_0}(\xi) \left( \sum_{l \in \mathbb{Z}} \widehat{\psi_0} (\xi+ l n_0) \overline{\widehat{g}(\xi+l n_0)}, \left\{ \widehat{\psi_0}(\xi+l n_0) \right\}_{l \in \mathbb{Z}} \right) \right\} \right) + \cdots  \\
    & \hspace{-1.2cm}+\Dim \left( \Span \left\{ \mathcal{X}_{H_{n_0-1}}(\xi+n_0-1) \left( \sum_{l \in \mathbb{Z}} \left(\widehat{\psi_{n_0-1}} \overline{\widehat{g}}\right)(\xi+n_0-1+ln_0), \left\{ \widehat{\psi_{n_0-1}}(\xi+ n_0 -1+ ln_0) \right\}_{l \in \mathbb{Z}} \right) \right\} \right)\Bigg).
    \end{align*}
    In the last statement, we have used the fact that for all $i \in \{0, \ldots, n_0-1\}$ and a.e. $\xi \in [0,1]$, $\mathcal{X}_{E_i} (\xi+ i +ln_0)= \mathcal{X}_{H_i}(\xi+i)$. Now, for a.e. $\xi \in [0,1]$, by definition of $H_0$, it follows that $ \xi \in H_0$ if and only if $ 0 \in D_\xi$. Similarly, for any $i \in \{0, \ldots, n_0-1\}$, $\xi +i \in H_{i}$ if and only if $ i \in D_\xi$. However, $\#D_\xi = l$, for a.e. $ \xi \in [0,1]$. Therefore, for a.e. $ \xi \in [0,1]$, only at the most $l$ of $\left\{\mathcal{X}_{H_i}(\xi+i) \right\}_{i=0}^{n_0-1}$ can survive, which along with \eqref{30} implies that $\Len W \leq l$.\\

    Now that we have proved $ W \in \mathcal{V}^l_{n_0}$, the second step is to show that it is a maximizer of \eqref{28}. \\
     It is enough to show that for a.e. $\xi \in [0,1]$,
    \begin{align*}
        \sum_{k=0}^{n_0-1} \sum_{j=1}^m \left\| P_{\widetilde{J}_{\widetilde{V_k}}(\xi+k)} \left( \widetilde{\Gamma} \widetilde{f^k_{Y,j}} \right) (\xi+k) \right\|^2 \leq \sum_{k=0}^{n_0-1} \sum_{j=1}^{m} \left\| P_{\widetilde{J}_{\widetilde{W_k}}(\xi+k)} \left( \widetilde{\Gamma} \widetilde{f^k_{Y,j}} \right) (\xi+k) \right\|^2.
    \end{align*}
    Let $A^V_{\xi}$ be as defined in \eqref{AVxi}. Then, for a.e. $\xi\in [0,1]$
\begin{align*}
    \sum_{k=0}^{n_0-1} \sum_{j=1}^m \left\|P_{\widetilde{J}_{\widetilde{V_k}}(\xi+k)} \left( \widetilde{\Gamma} \widetilde{f^k_{Y,j}} \right) (\xi+k) \right\|^2 &= \sum_{k \in A^V_{\xi}} \sum_{j=1}^m \left\| P_{\widetilde{J}_{\widetilde{V_k}} (\xi+k)} \left( \widetilde{\Gamma} \widetilde{f^k_{Y,j}} \right) (\xi+k) \right\|^2\\
    \numberthis \label{31} &=\sum_{k \in A^V_\xi \cap D_\xi} \sum_{j=1}^m 
 \left\|P_{\widetilde{J}_{\widetilde{V_k}}(\xi+k)} \left( \widetilde{\Gamma} \widetilde{f^k_{Y,j}} \right) (\xi+k) \right\|^2 \\
 &+ \sum_{k \in A^V_\xi \setminus (A^V_\xi  \cap D_\xi)} \sum_{j=1}^m \left\| P_{\widetilde{J} _{\widetilde{V_k}}(\xi+k)} \left( \widetilde{\Gamma} \widetilde{f^{k}_{Y,j}} \right) (\xi+k)\right \|^2.
\end{align*}
We consider two cases. 
\begin{enumerate}
    \item Let $ k \in A^V_{\xi} \cap D_\xi$. Then by \eqref{comparisonforUwithoutintegral}, 
\begin{align*}
    \sum_{j=1}^m \left\| P_{\widetilde{J}_{\widetilde{V_k}}(\xi+k)} \left( \widetilde{\Gamma} \widetilde{f^k_{Y,j}} \right) (\xi+k) \right\|^2& \leq \sum_{j=1}^m \left\| P_{\widetilde{J}_{\widetilde{U_k}}(\xi+k)} \left(\widetilde{\Gamma} \widetilde{f^k_{Y,j}} \right) (\xi+k) \right\|^2\\
    \numberthis \label{32} &= \sum_{j=1}^m \left\| P_{\widetilde{J}_{\widetilde{W_k}}(\xi+k)} \left( \widetilde{\Gamma} \widetilde{f^k_{Y,j}} \right) (\xi+k) \right\|^2.
\end{align*}
The equality \eqref{32} is obtained using the fact that $ k \in D_\xi$, which  inturn implies that $\xi+k \in H_k$, which further in turn implies that  $\widehat{\psi_k} (\xi+k +n_0j)= \widehat{\phi_k} (\xi+k+n_0j)$ for all $j \in \mathbb{Z}$. From this, we can conclude that $\widetilde{J}_{\widetilde{W_k}}(\xi+k)= \widetilde{J}_{\widetilde{U_k}}(\xi+k)$. Summing over $k \in A_\xi \cap D_\xi$ in inequality \eqref{32}, we get
\begin{equation}\label{33}
    \sum_{k \in A^V_\xi \cap D_\xi} \sum_{j=1}^m \left\| P_{\widetilde{J}_{\widetilde{V_k}}(\xi+k)} \left( \widetilde{\Gamma} \widetilde{f^k_{Y,j}} \right) (\xi+k) \right\|^2 \leq \sum_{k \in A^V_\xi \cap D_\xi} \sum_{j=1}^m \left\| P_{\widetilde{J}_{\widetilde{W_k}}(\xi+k)} \left( \widetilde{\Gamma} \widetilde{f^k_{Y,j}} \right) (\xi+k) \right\|^2.
\end{equation}
\item Let $ k \in A_\xi \setminus (A_\xi \cap D_\xi)$. Again using  \eqref{comparisonforUwithoutintegral}, we have 
\[\sum_{j=1}^m \left\| P_{\widetilde{J}_{\widetilde{V_k}}(\xi+k)} \left( \widetilde{\Gamma} \widetilde{f^k_{Y,j}} \right) (\xi+k) \right\|^2 \leq \sum_{j=1}^m \left\| P_{\widetilde{J}_{\widetilde{U_k}}(\xi+k)} \left( \widetilde{\Gamma} \widetilde{f^k_{Y,j}} \right) (\xi+k) \right\|^2.\]
Choose an $l \in D_xi$, then using the fact that $k \notin D_\xi$ and $ l \in D_\xi$, we get
\begin{align*}
    \sum_{j=1}^m \left\| P_{\widetilde{J}_{\widetilde{U_k}}(\xi+k)} \left( \widetilde{\Gamma} \widetilde{f^k_{Y,j}} \right) (\xi+k) \right\|^2 & \leq \sum_{j=1}^m \left\| P_{\widetilde{J}_{\widetilde{U_l}} (\xi+l)} \left( \widetilde{\Gamma} \widetilde{f^l_{Y,j}} \right) (\xi+l) \right\|^2\\
\numberthis \label{34} &= \sum_{j=1}^m \left\| P_{\widetilde{J}_{\widetilde{W_l}}(\xi+l)} \left( \widetilde{\Gamma} \widetilde{f^l_{Y,j}} \right) (\xi+l) \right\|^2.
\end{align*}
Now, $\#A^V_\xi \leq l =\#D_\xi$. This implies that $\#\left(A^V_\xi \setminus (A^V_\xi \cap D_\xi)\right) \leq \#\left(D_\xi \setminus (A^V_\xi \cap D_\xi)\right)$, which means that for each distinct $ k \in A^V_\xi \setminus(A^V_\xi \cap D_\xi)$, we  find a distinct $l \in D_\xi \setminus (A^V_\xi \setminus D_\xi)$ satisfying \eqref{34}. Hence,
\begin{equation}\label{35}
\sum_{k \in A^V_\xi \setminus (A^V_\xi \cap D_\xi)} \sum_{j=1}^m \left\|P_{\widetilde{J}_{\widetilde{V_k}}(\xi+k)} \left( \widetilde{\Gamma} \widetilde{f^k_{Y,j}} \right) (\xi+k) \right\|^2 \leq \sum_{k \in D_\xi \setminus (A^V_\xi \cap D_\xi)}\sum_{j=1}^m \left\| P_{\widetilde{J}_{\widetilde{W_k}}(\xi+k)} \left( \widetilde{\Gamma} \widetilde{f^k_{Y,j}}   \right)(\xi+k) \right\|^2.
\end{equation}
\end{enumerate}

Hence, finally, using \eqref{33} and \eqref{35}, we can conclude that 
 \begin{align*}
     \sum_{k=0}^{n_0-1} \sum_{j=1}^m \left\| P_{\widetilde{J}_{\widetilde{V_k}}(\xi+k)} \left( \widetilde{\Gamma} \widetilde{f^k_{Y,j}} \right) (\xi+k) \right\|^2&= \sum_{k \in A^V_\xi}\sum_{j=1}^m \left\| P_{\widetilde{J}_{\widetilde{V_k}}(\xi+k)} \left( \widetilde{\Gamma} \widetilde{f^k_{Y,j}} \right) (\xi+k) \right\|^2\\
     & \leq \sum_{k \in D_\xi} \sum_{j=1}^m \left\| P_{\widetilde{J}_{\widetilde{W_k}}(\xi+k)} \left( \widetilde{\Gamma} \widetilde{f^k_{Y,j}} \right) (\xi+k) \right\|^2\\
     &= \sum_{k=0}^{n_0-1}\sum_{j=1}^m \left\| P_{\widetilde{J}_{\widetilde{W_k}}(\xi+k)} \left( \widetilde{\Gamma} \widetilde{f^k_{Y,j}} \right) (\xi+k) \right\|^2.
 \end{align*}
The last equality follows from that fact that if $k\notin D_\xi$, then by the definition of $\psi_k,$ $\widetilde{J}_{\widetilde{W_k}}(\xi+k)=\{\boldsymbol{0}\}$. Hence our claim is proved.
 \end{proof}

\section{Approximation with Paley Wiener spaces}\label{palaeyr}

Fix $l \in \mathbb{N}$.  Define the space \cite{cabrelliacha2016}
\begin{align*}
    \mathcal{T}^l= \bigg\{ V=V(\phi_1,\ldots,\phi_l): \,& \phi_1, \ldots, \phi_l \in L^2(\mathbb{R}), V \text{ is translation invariant and }  \\
    & \left\{ T_k\phi_i: k \in \mathbb{Z}, i \in \{1, \ldots,l \} \right\} \text{ forms a Riesz basis for } V \bigg\} .
\end{align*}
Given measurements $\left\{Y^j \right\}_{j=1}^m = \left\{ \left\{ y^1_k \right\}_{k \in \mathbb{Z}}, \ldots, \left\{ y^m_k \right\}_{k \in \mathbb{Z}} \right\} \subset \ell^2(\mathbb{Z})$, we want to solve the minimization problem (we make use of the first form, see \eqref{8}) 
\[\argmin_{V \in \mathcal{T}^l} \sum_{j=1}^m \left\| Y^j- P_{\widetilde{V}} Y^j \right\|^2.\]

In fact, we shall minimize over a smaller collection $\mathcal{T}^l_N$ (defined in \eqref{TlN}), which approximates $\mathcal{T}^l$.\smallskip 
% Given a collection $\left\{ Y^j=\left\{ y^j_k \right\}_{k \in \mathbb{Z}} \right\}_{j=1}^m \subset \ell^2(\mathbb{Z})$, we want to solve
%  $\argmin_{V \in \mathcal{T}^l} \sum_{j=1}^m \left\| Y^j - P_{\widetilde{V}}Y^j \right\|^2.$    \\

Further from Wiener's theorem, we know that $V$ is a translational invariant subspace of $L^2(\mathbb{R})$ if and only if there exists a measurable set $ \Omega \subset \mathbb{
R}$ such that 
\[ V= \left\{ f \in L^2(\mathbb{R}): \widehat{f}(\xi)=0, \text{ for a.e. } \xi \in \mathbb{R} \setminus  \Omega \right\}.\]

We denote $V=V_\Omega$ (as $\Omega$ is unique upto measure zero).
\begin{definition}\cite{cabrelliacha2016}
    Let $\Omega \subset \mathbb{R}$ be measurable and $l \in \mathbb{N}$. We say that $\Omega$ $l$ multi-tiles $\mathbb{R}$ if 
\[\sum_{k \in \mathbb{Z}} \mathcal{X}_\Omega (\xi -k) =l, \text{  for a.e. } \xi \in \mathbb{R}.\]
\end{definition}
% \begin{lemma}
%     A measurable set $\Omega \subset \mathbb{R}$, $l$ multi-tiles $\mathbb{R}$ if and only if 
% \[\Omega=\Omega_1 \cup \cdots \cup \Omega_l \cup N,\]
% where $N$ is a zero measure set and the sets $\Omega_j$, $1 \leq j \leq l$ are measurable, disjoint and each of them tiles $\mathbb{R}$ by translation on $\mathbb{Z}$.
% \end{lemma}
\begin{prop}\cite[Proposition 4.3]{cabrelliacha2016}\label{mainlemmafromcabrellipaper}
    A subspace $V$ is in $\mathcal{T}^l$ if and only if $V=V_\Omega$ for some $\Omega$ a measurable $l$ multi-tile of $\mathbb{R}.$ Moreover, in such a case, $\Dim \left((J_V(\xi) \right)=l$, for a.e. $\xi \in [0,1]$ .
\end{prop}
% For a.e. $\xi \in [0,n_0]$, let 
% \begin{enumerate}
%     \item $O^\Omega_\xi:= \left\{ k \in \mathbb{Z}: \xi+ kn_0 \in \Omega  \right\},$ and\\
%     \item $S \left( O^\Omega_\xi \right):= \overline{\Span} \left\{ \left( \overline{\hat{g}(\xi+kn_0)}, e_k \right): k \in O^\Omega_\xi \right\}.$
% \end{enumerate}
\begin{definition}
    Let $n_0 \in \mathbb{N}$ be the assumed measurement rate. For a.e. $\xi \in [0,n_0]$ and any $\Omega \subset \mathbb{R}$, define
    \begin{enumerate}
        \item $O^\Omega_\xi:= \left\{ k \in \mathbb{Z}: \xi + k n_0 \in \Omega \right\}$ and 
        \item $S\left( O^\Omega_\xi \right)= \overline{\Span} \left\{ \left(\overline{\hat{g}(\xi+kn_0)}, e_k \right): k \in O^\Omega_\xi \right\}.$
    \end{enumerate}
\end{definition}
\begin{lemma}\label{mainlemma}
    Let $V=V_\Omega \in \mathcal{T}^l$. Then, $\widetilde{J}_{\widetilde{V_\Omega}}(\xi) \cong S\left( O^\Omega_\xi \right)$, for a.e. $\xi \in [0,n_0]$.
\end{lemma}
% \begin{lemma}\label{mainlemma}
%     Let $V=V_\Omega \in \mathcal{T}^l$. Then, $\widetilde{J}_{\widetilde{V_\Omega}}(\xi)\cong S\left(O^\Omega_\xi \right)$, for a.e. $\xi \in [0,n_0].$
% \end{lemma}
\begin{proof}
    % Let $\xi \in [0,n_0]$. Then since  $\widetilde{J}_{\widetilde{V_\Omega}}(\xi)$ is defined as $$\widetilde{J}_{\widetilde{V_\Omega}}(\xi)= \left\{ \left( \sum_{k \in \mathbb{Z}} f(\xi + kn_0) \overline{\widehat{g}(\xi+ k n_0)}, \left\{ \widehat{f}(\xi+ k n_0) \right\}_{k \in \mathbb{Z}} \right): f\in V_\Omega \right\},$$ by definition of $O_\xi$, we clearly have $\widetilde{J}_{\widetilde{V_\Omega}}(\xi) \subset S\left(O^\Omega_\xi \right)$. For the converse, fix $\xi \in [0,n_0]$: Let $E_k=\left( [0,n_0] +kn_0 \right) \cap \Omega \quad \forall \hspace{0.1cm} k \in \mathbb{Z}$. Then, we have $\Omega= \cup_{k \in \mathbb{Z}}E_k$ and $k \in O^\Omega_\xi$ if and only if $\xi+ k n_0 \in E_k$. Hence, if $a= \left( \sum_{k \in \mathbb{Z}} a_k \overline{\widehat{g}(\xi + k n_0)}, \left\{ a_k \right\}_{k \in \mathbb{Z}} \right) \in S \left(O^\Omega_\xi \right)$, then $H_\xi(x)= \sum_{k \in \mathbb{Z}} a_k\mathcal{X}_{E_k} (x) \in L^2(\Omega)$. That is, $\widehat{h}=H_\xi \in V_\Omega$ and $\widehat{h}(\xi + k n_0) =H_\xi (\xi+ k n_0)=a_k$ if $k \in O^\Omega_\xi$. Therefore, $h \in V$ and $\left( \widetilde{\Gamma} \widetilde{h} \right)(\xi)= \left( \sum_{k \in \mathbb{Z}} \widehat{h}(\xi+ k n_0) \overline{\widehat{g}(\xi+ k n_0)}, \left\{ \widehat{h}(\xi+ k n_0) \right\}_{k \in \mathbb{Z}} \right)=a \in \widetilde{J}_{\widetilde{V_\Omega}}(\xi)$.

Recall that for any $\xi \in [0,n_0]$,  $\widetilde{J}_{\widetilde{V_\Omega}}(\xi)$ is defined as 
$$\widetilde{J}_{\widetilde{V_\Omega}}(\xi)= \left\{ \left( \sum_{k \in \mathbb{Z}} f(\xi + kn_0) \overline{\widehat{g}(\xi+ k n_0)}, \left\{ \widehat{f}(\xi+ k n_0) \right\}_{k \in \mathbb{Z}} \right): f\in V_\Omega \right\}.$$
Hence, by definition of $O^\Omega_\xi$, it follows that $\widetilde{J}_{\widetilde{V_\Omega}}(\xi) \subset S \left( O^\Omega_\xi \right)$. Now we prove the converse. Fix $\xi \in [0,n_0]$ and let $E_k= \left( [0,n_0] + kn_0 \right) \cap \Omega \hspace{0.2cm} \forall \hspace{0.1cm} k \in \mathbb{Z}$. Then, using the fact that $\Omega = \cup_{k \in \mathbb{Z}} E_k$, we  can show that $k \in O^\Omega_\xi$ if and only if $\xi+ k n_0 \in E_k$. Consider any $a= \left( \sum_{k \in \mathbb{Z}} a_k \overline{\widehat{g}(\xi + k n_0)}, \left\{ a_k \right\}_{k \in \mathbb{Z}} \right) \in S \left(O^\Omega_\xi \right)$, then the function $H_\xi(x):= \sum_{k \in \mathbb{Z}} a_k\mathcal{X}_{E_k} (x) $ belongs to $ L^2(\Omega)$. That is, $\widehat{h}=H_\xi \in V_\Omega$. Further, if $k \in O^\Omega_\xi$, then $\widehat{h}(\xi + k n_0) =H_\xi (\xi+ k n_0)=a_k.$ Therefore, $h \in V$ and $\left( \widetilde{\Gamma} \widetilde{h} \right)(\xi)= \left( \sum_{k \in \mathbb{Z}} \widehat{h}(\xi+ k n_0) \overline{\widehat{g}(\xi+ k n_0)}, \left\{ \widehat{h}(\xi+ k n_0) \right\}_{k \in \mathbb{Z}} \right)=a \in \widetilde{J}_{\widetilde{V_\Omega}}(\xi)$.
\end{proof}
From Proposition \ref{mainlemmafromcabrellipaper}, it is clear that in order to find an optimal subspace in the class $\mathcal{T}_l$, it is enough to find the associated $l$ multi-tile $\Omega$  in $ \mathbb{R}$. As in \cite{cabrelliacha2016}, we restrict $\Omega$ to be inside a cube that may be arbitrarily large.  
\begin{definition}\cite{cabrelliacha2016} 
Let $N \in \mathbb{N}$.
 Define  
 \begin{enumerate}
\item $C_N=\left[ - \left(N+ \frac{1}{2} \right), N+ \frac{1}{2} \right].$
\item $M^l_N= \{ \Omega \subset C_N: \Omega \text{ is measurable and } l  \text{ multi-tiles } \mathbb{R}\}$.
\item 
\begin{equation} \label{TlN}
\mathcal{T}^l_N= \left\{ V \in \mathcal{T}^l : V= V_\Omega \text{ with } \Omega \in M^l_N \right\}.
\end{equation}
\end{enumerate}
\end{definition}
% Define,
% \begin{align*}
%     C_N&=[-N,N],\\
%     M^l_N&=\left\{ \Omega \subset C_N: \Omega \text{ is measurable and } l \text{ multi-tiles } \mathbb{R} \right\}\\
% \mathcal{T}^l_N&= \left\{ V \in \mathcal{T}^l: V= V_\Omega \text{ with } \Omega \in M^l_N \right\}.
% \end{align*}
We now state the main result of this section.
\begin{theorem}
    Let $l \in \mathbb{N}$. Suppose the measurements $\left\{Y^j \right\}_{j=1}^m= \left\{ \left\{y^1_k \right\}_{\in \mathbb{Z}}, \ldots, \left\{ y^m_k \right\}_{k \in \mathbb{Z}} \right\} \subset \ell^2(\mathbb{Z})$ are given. Then for each $N \geq l$, there exists a Paley Wiener space $V^* \in \mathcal{T}^l_N$ that satisfies 
    \begin{equation}\label{paleywienerminproblem}
    V^*= \argmin_{V \in \mathcal{T}^l_N} \sum_{j=1}^m \left\|  Y^j- P_{\widetilde{V}} Y^j \right\|^2 \text{(Minimization Problem Form 1)}. 
    \end{equation}
\end{theorem}
\begin{proof}
   If $V^*$ exists, then $V^*= \argmin_{V \in \mathcal{T}^l_N} \sum_{j=1}^m \left\|Y^j- P_{\widetilde{V}} Y^j \right\|^2= \argmax_{V \in \mathcal{T}^l_N} \sum_{j=1}^m \left\| P_{\widetilde{V}} Y^j \right\|^2$.  Further, $\max_{V \in \mathcal{T}^l_N} \sum_{j=1}^m \left\|P_{\widetilde{V}} Y^j \right\|^2= \max_{\Omega \in M^l_N} \sum_{j=1}^m \left\| P_{\widetilde{V_\Omega}} Y^j \right\|^2 .$ For each $k \in \{0, \ldots, n_0-1\}$, let $Y^{j,k}$ be defined as the element in $\mathcal{R}_\lambda$ satisfying 
\[ \left( \widetilde{\Gamma} Y^{j,k}\right)(\xi)= \mathcal{X}_{[k,k+1]}(\xi) \left( \widetilde{\Gamma} Y^{j} \right)(\xi), \text{ for a.e. } \xi \in [0,n_0].\]
Furthermore, we decompose $V_\Omega$ as the orthogonal direct sum $V_\Omega= V_{\Omega,0} \dot{\oplus} \cdots \dot{\oplus} V_{\Omega, n_0-1}$(see \eqref{101}). From the definition of $V_{\Omega,k}$, it is clear that $Y^{j,k} \perp \widetilde{V_{\Omega,l}}$, for all $k \neq l \in \{0, \ldots, n_0-1 \}.$ Indeed, $\widetilde{\Gamma}Y^{j,k} \perp \widetilde{\Gamma}\left(\widetilde{V_{\Omega,l}}\right)$, for all $k \neq l \in \{0, \ldots, n_0-1 \}.$  Thus, using Proposition \ref{307}, we get 
\begin{align*}
    \sum_{j=1}^m \left\| P_{\widetilde{V_\Omega}} Y^j \right\|^2&= \sum_{j=1}^m \sum_{k=0}^{n_0-1} \left\| \widetilde{\Gamma} P_{\widetilde{V_{\Omega,k}}} Y^{j,k} \right\|^2\\
    &= \sum_{j=1}^m \sum_{k=0}^{n_0 -1} \int_0^{n_0} \left\| \left( \widetilde{\Gamma}P _{\widetilde{V_{\Omega,k}}} Y^{j,k} \right)(\xi) \right\|^2 d\xi\\
    &= \sum_{j=1}^m \sum_{k=0}^{n_0-1} \int_{k}^{k+1} \left\| \left( \widetilde{\Gamma}P_{\widetilde{V_{\Omega,k}}} Y^{j,k} \right) (\xi) \right\|^2 d\xi\\
    &= \sum_{j=1}^m \sum_{k=0}^{n_0-1} \int_0^1 \left\| \left( \widetilde{\Gamma} P_{\widetilde{V_{\Omega,k}}} Y^{j,k} \right) (\xi+k) \right\|^2 d\xi\\
    &= \sum_{j=1}^m\sum_{k=0}^{n_0-1} \int_0^1
 \left\| P_{\widetilde{J}_{\widetilde{V_{\Omega,k}}}(\xi+k)} \left( \widetilde{\Gamma} Y^{j,k} \right) (\xi+k)\right\|^2 d\xi.
\end{align*}

Notice that $V_{\Omega,k} = V_{\Omega \cap B_k}$. This implies that $V_{\Omega,k} \in \mathcal{T}^l$. For ease of notation, let $\Omega_k:= \Omega \cap B_k$ for all $k \in \{0, \ldots, n_0-1 \}$. Then, using Lemma \ref{mainlemma}, we get 
\begin{equation}\label{Jomega=S(omega)}
\widetilde{J}_{\widetilde{V_{\Omega_k}}}(\xi)= S \left( O^{\Omega_k}_\xi \right), \text{ for a.e. } \xi \in [0,n_0].
\end{equation}
If $\Omega \in M^l_N$, then it follows from  Proposition \ref{mainlemmafromcabrellipaper} and  \ref{mainlemmafromcabrellipaper} that $\Dim \left( J_V(\xi) \right)=l$, for a.e. $\xi \in [0,1]$. Further using the calculations done in order to arrive at \eqref{newdimensionformula} and making use of Proposition 
 and \eqref{Jomega=S(omega)}, for a.e. $\xi \in [0,1]$, we get
\begin{align*}
 l &= \Dim \left( J_V(\xi) \right)= \Dim \widetilde{J}_{\widetilde{V_0}}(\xi)+ \cdots+ \Dim \widetilde{J}_{\widetilde{V_{\Omega_{n_0-1}}}}(\xi+ n_0-1)\\
 \numberthis \label{l=dim(S(O))} &= \Dim S\left( O^{\Omega_0}_\xi \right) + \cdots+ \Dim S \left( O^{\Omega_{n_0-1}}_{\xi + n_0-1} \right).
 \end{align*}
 Let $ l^\xi_k :=\#O^{\Omega_k}_{\xi+k}$ for all $ k \in \{ 0, \ldots, n_0-1 \}$ and a.e. $\xi \in [0,1]$. Then, we can conclude from \eqref{l=dim(S(O))} that $l^0_\xi + \cdots + l^{n_0-1}_\xi=l$, for a.e. $\xi \in [0,1]$. Further, there exists a unique set of $l^k_\xi$ integers $\left\{ r_1(\xi,k), \ldots, r_{l^k_\xi} (\xi,k) \right\} $ such that 
 \begin{equation}\label{finalJtilde}
\widetilde{J}_{\widetilde{V_{\Omega_k}}}(\xi+k) =S\left( O^{\Omega_k}_{\xi+k} \right)=\Span \left\{ \left( \overline{\hat{g}(\xi+k+n_0 r_i(\xi,k))}, e_{r_i(\xi,k)} \right): i \in \left\{ 1, \ldots, l^k_\xi \right\} \right\}.
\end{equation}
Since $\Omega \subset C_N, \left|k + n_0 r_i(\xi,k) \right| \leq N$ for all $i \in \left\{1, \ldots, l^k_\xi \right\}$ ,  a.e. $\xi \in [0,1]$ and  all $k \in \{0, \ldots,n_0-1 \}$. Using the above observations and the definition of $\widetilde{\Gamma}$, we get
\begin{align*}
\sum_{j=1}^m \sum_{k=0}^{n_0-1} \int_0^1 \left\| P_{\widetilde{J}_{\widetilde{V_{\Omega_k}}}(\xi+k)} \left( \widetilde{\Gamma} \widetilde{Y^{j,k}} \right)(\xi+k) \right\|^2d\xi&= \sum_{j=1}^m \sum_{k=0}^{n_0-1} \int_0^1 \left\| P_{S\left( O^{\Omega_k}_{\xi+k} \right)} \widetilde{\Gamma} \widetilde{Y^{j,k}} (\xi+k) \right\|^2 d\xi\\
&\hspace{-2cm}= \sum_{j=1}^m \sum_{j=0}^{n_0-1} \int_0^1 \left\| P_{S \left( O^{\Omega_k}_{\xi+k} \right)} \left( \mathcal{X}_{[k,k+1]}(\xi+k) \left( \sum_{l \in \mathbb{Z}} y^j_l e^{- \frac{2 \pi i (\xi+k)l}{n_0}}, \boldsymbol{0} \right) \right) \right\|^2d\xi\\
& \hspace{-2cm}=\sum_{j=1}^m \sum_{k=0}^{n_0-1} \int_0^1 \left\| P_{S \left( O^{\Omega_k}_{\xi+k} \right)} \left( \sum_{l \in \mathbb{Z}}y^j_l e^{- \frac{2 \pi i (\xi+k)l}{n_0}} \left( 1, \boldsymbol{0} \right) \right) \right\|^2d\xi\\
& \hspace{-2cm}=\sum_{j=1}^m \sum_{k=0}^{n_0-1} \int^1_0 \left| \sum_{l \in \mathbb{Z}} y^j_l e^{- \frac{2 \pi i (\xi+k)l}{n_0}} \right|^2 \left\| P_{S\left( O^{\Omega_k}_{\xi+k} \right)} \left( 1, \boldsymbol{0} \right) \right\|^2 d\xi.
\end{align*}
Given $\Omega \in M^l_N$, for a.e. $\xi \in [0,1]$, the set $\Omega$ contains exactly $l$ elements from the sequence $\{ \xi+k: k \in \mathbb{Z} \}$. Therefore, we define $S$, (the set of possible translations) to be 
\begin{align*}
S= \bigg\{\boldsymbol{s}= \left(s_0, \ldots, s_{n_0-1} \right)= \left( \left\{s^1_0, \ldots, s^{l_0}_0 \right\}, \ldots, \left\{ s^1_{n_0-1}, \ldots, s^{l_{n_0-1}}_{n_0-1} \right\} \right):\,& l_0+ \cdots+l_{n_0-1}=l,\\
&\forall\, i \in \{0, \ldots, n_0-1 \}, s_i \subset \mathbb{Z},\\
&  \text{the integers contained in } s_i \text{ are distinct, } \\
& \text{and }\left\|n_0s_i+i \right\|_{\infty} \leq N \bigg\}.
\end{align*}
For every $ i \in \{0, \ldots,n_0-1 \}$, $n_0s_i+i:= \{ n_0s^i_j+i \}_{j=1}^{l_i}$. In the above definition, $l_i$ can take the value $0$, in which case $s_i:= \phi$ (the empty set). Additionally, for each $\boldsymbol{s} \in S $, let 
\begin{equation*}
 F_{s_k}(\xi):= \left\| P_{S\left( O^{s_k}_{\xi+k} \right)}
 (1, \boldsymbol{0}) \right\|^2 \text{ for all } k\in \{0, \ldots,n_0-1\} \text{ and  a.e. } \xi \in [0,1].
\end{equation*}
Notice that the space $S\left( O^{s_k}_{\xi+k} \right)$ is very similar to the space $A_\eta$ defined in \eqref{42}. Hence, we use the methods developed in Subsection \ref{fYj} in order to calculate $P_{S \left(O^{s_k}_{\xi+k} \right)} (1, \boldsymbol{0} ).$
\smallskip
Fix $k \in \{ 0, \ldots, n_0-1\}$ and $\xi \in [0,1]$. Now, for all $i \in \left\{1, \ldots, l_k \right\}$, we define

\[\widetilde{e}_i:= e_{s^i_k} \text{ and } a_i:= \left\langle \widetilde{e}_i, g_{\xi+k} \right \rangle. \]
Here, as in Subsection \ref{fYj}, $g_\eta:= \left\{ \hat{g}(\eta+l n_0) \right\}_{l \in \mathbb{Z}}$, for a.e. $\eta \in [0,n_0]$. Note that both $a_i$ and $\widetilde{e}_i$ are implicitly dependent on $\xi$ and $k$. As the collection $\left\{ \left(a_i, \widetilde{e}_i \right) \right\}_{i=1}^{l_k}$ forms a Riesz basis for $S\left( O^{s_k}_{\xi+k} \right)$, we orthonormalize it using the Gram-Schmidt orthonormalization process  to get the orthonormal basis $\left\{ \frac{v_i}{\left\|v_i \right\|} \right\}_{i=1}^{l_k}$ as defined in \eqref{vn}. Therefore,

\begin{equation*}
F_{s_k}(\xi)=\left\| P_{S\left(O^{s_k}_{\xi+k} \right)}(1, \boldsymbol{0}) \right\|^2= \left\| \sum_{i=1}^{l_k} \frac{a_i}{\left( |a_i|^2+ \cdots+ |a_0|^2+ \lambda \right)}v_n\right\|^2.
\end{equation*}

In order to find the optimal space $V_{\Omega^*}$, we construct $\Omega^*$ using the following strategy. For a.e. $\xi \in [0,1]$, we pick $\boldsymbol{s}^*= \left( s^*_0, \ldots, s^*_{n_0-1} \right) \in S$ such that $\sum_{j=1}^m \sum_{k=0}^{n_0-1} \left| \sum_{l \in \mathbb{Z}}y^j_l e^{- \frac{2 \pi i (\xi+k)l}{n_0}} \right|^2F_{s^*_k}(\xi)$ is maximum taken over all $\boldsymbol{s} \in S$
. The maximum exists because for a.e. $\xi \in [0,1]$ and all $\boldsymbol{s} \in S$, $$\sum_{j=1}^{m} \sum_{k=0}^{n_0-1}\left| \sum_{l \in \mathbb{Z}} y^j_l e^{- \frac{2 \pi i (\xi+k)l}{n_0}} \right|^2 F_{s_k}(\xi)  < \infty \text{ and additionally } \#S< \infty.$$  Rigorously, we define $\Omega^*$ as follows. For each $\boldsymbol{s} \in S$, let
\[E_{\boldsymbol{s}}:= \left\{ \xi \in [0,1]: \sum_{j=1}^m \sum_{k=0}^{n_0-1} \left| \sum_{l \in \mathbb{Z}}y^j_l e^{- \frac{2 \pi i (\xi+k)l}{n_0}} \right|^2F_{s_k}(\xi) \geq \sum_{j=1}^m \sum_{k=0}^{n_0-1} \left| \sum_{l \in \mathbb{Z}}y^j_l e^{- \frac{2 \pi i (\xi+k)l}{n_0}} \right|^2F_{r_k}(\xi), \hspace{0.2cm} \forall \hspace{0.1cm} \boldsymbol{r} \in S \right\}.\]
Finally, let
\[ \Omega^*:= \cup_{\boldsymbol{s} \in S} \cup_{k=0}^{n_0-1} \cup_{i=0}^{l_k} \left( E_{\boldsymbol{s}} +k + n_0 s^i_k \right).\]
Clearly, from its definition, $E_{\boldsymbol{s}}$ is a measurable set for each $\boldsymbol{s} \in S$, which in turn implies that $\Omega^*$ is measurable. Further, by construction $\Omega^* \in M^l_N$. Let $\Omega \in M^l_N$ be arbitrary. Then, for a.e. $\xi \in [0,1]$, we have 
\[\sum_{j=1}^m \sum_{k=0}^{n_0-1} \left| \sum_{l \in \mathbb{Z}}y^j_l e^{- \frac{2 \pi i (\xi+k)l}{n_0}} \right|^2 \left\| P_{S \left( O^{\Omega_k}_{\xi+k} \right)} (1, \boldsymbol{0}) \right\|^2 \leq \sum_{j=1}^m \sum_{k=0}^{n_0-1} \left| \sum_{l \in \mathbb{Z}} y^j_l e^{- \frac{2 \pi i (\xi+k)l}{n_0}} \right|^2 F_{s^*_k}(\xi).
\]
Taking integral over $\xi \in [0,1]$, we get 
\begin{align*}
 \sum_{j=1}^m\left\| P_{\widetilde{V_\Omega}}Y^j \right\|^2&= \sum_{j=1}^m \sum_{k=0}^{n_0-1} \int^1_0 \left| \sum_{l \in \mathbb{Z}} y^j_l e^{- \frac{2 \pi i (\xi+k)l}{n_0}} \right|^2 \left\| P_{S\left( O^{\Omega_k}_{\xi+k} \right)} (1, \boldsymbol{0}) \right\|^2d \xi\\
 &\leq \sum_{j=1}^m \sum_{k=0}^{n_0-1}\int_0^1 \left| \sum_{l \in \mathbb{Z}}y^j_l e^{- \frac{2 \pi i (\xi+k)l}{n_0}} \right|^2 F_{s^*_k}(\xi)d \xi\\
 & = \sum _{j=1}^m \left\| P_{\widetilde{V_{\Omega^*}}}Y^j \right\|^2.
 \end{align*}
Hence, we can conclude that $V_{\Omega^*} \in \mathcal{T}^l_N$ is a solution of \eqref{paleywienerminproblem}.
\end{proof}
\section{Acknowledgement}
We thank Prof. Akram Aldroubi for going through the manuscript and providing valuable suggestions that helped us improve the presentation of the paper from an earlier version to the current one.

\begingroup
\let\itshape\upshape
 
\bibliographystyle{plain}
	\bibliography{references}

\end{document}